\PassOptionsToPackage{final}{graphicx}
\PassOptionsToPackage{final}{hyperref}
\documentclass[acmsmall,nonacm,screen,natbib=false]{acmart}
\def\maintexname{goodstein-pldi} %

\usepackage{substr}
\usepackage{ifthen}
\usepackage{xr}
\newboolean{showmaintext}
\newboolean{showappendix}
\setboolean{showmaintext}{true} %
\setboolean{showappendix}{true} %
\newboolean{diffmode}
\setboolean{diffmode}{false}

\IfSubStringInString{\detokenize{-submission}}{\jobname}{
\setboolean{showmaintext}{true}
\setboolean{showappendix}{false}
}{}%
\IfSubStringInString{\detokenize{-appendix}}{\jobname}{
\setboolean{showmaintext}{false}
\setboolean{showappendix}{true}
}{}%

\numberwithin{equation}{section}
\usepackage{fontawesome}
\usepackage{tikz}
\usepackage{mathtools}
\usepackage{thmtools}
\usepackage{xr}
\usepackage{coqsrcref}

\usepackage{amssymb}
\usepackage{csquotes}
\usepackage{etoolbox}
\usepackage{aliascnt}
\usepackage{catchfilebetweentags}
\usepackage{agda}

\usepackage{enumitem}
\setlist[enumerate,1]{label=(\arabic*),font=\normalfont,align=left,leftmargin=0pt,labelindent=0pt,listparindent=\parindent,labelwidth=0pt,itemindent=!,topsep=2pt,parsep=0pt,itemsep=2pt,start=1}
\setlist[enumerate,2]{label=(\alph*),font=\normalfont,labelindent=*,leftmargin=*,start=1}
\setlist[itemize]{labelindent=*,leftmargin=*}
\setlist[description]{labelindent=*,leftmargin=*,itemindent=-1 em}

\newcommand{\nicehref}[2]{%
\begin{tikzpicture}[baseline=(txt.base)]
\node[text=blue!80!white!40!black,inner sep=0pt,text depth=1.1pt] (txt)
  {\href{#1}{{#2}{\hspace*{1pt}\raisebox{.0pt}{\color{blue!30!white}\tiny\faExternalLink}}}};
\begin{scope}[on background layer]
\draw[color=blue!25!white] (txt.south west) -- (txt.south east);
\end{scope}
\end{tikzpicture}%
}

\newcommand{\onlineHtmlURL}{%
  https://arxiv.org/src/2512.10748v3/anc/html/%
}

\newsavebox{\logoagdabox}
\sbox{\logoagdabox}{%
  \raisebox{-2pt}{\includegraphics[height=1em]{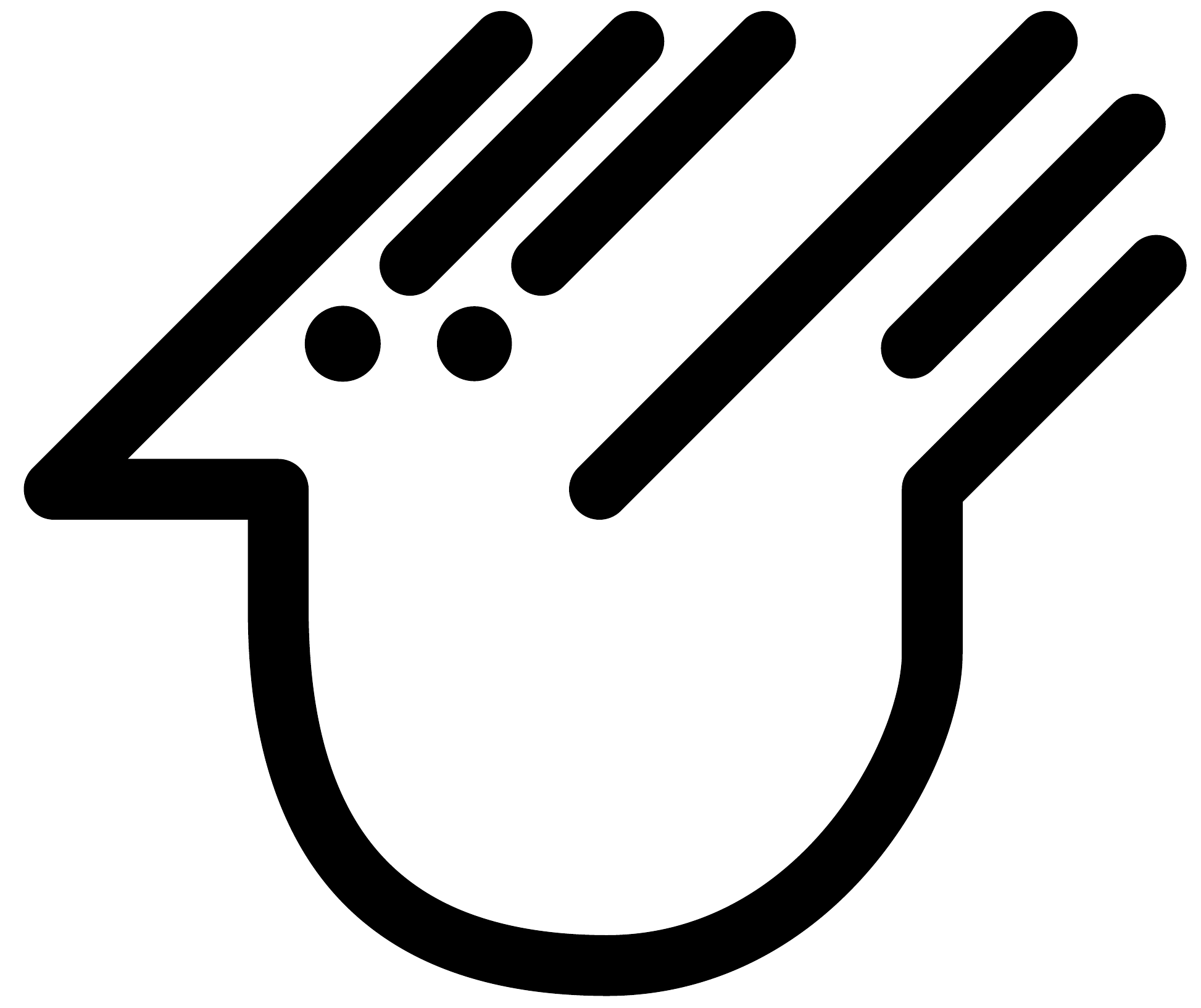}}%
}

\newcommand{\agdaref}[4][]{%
  \coqref{#2}{#3}{#4%
  }{\href{\onlineHtmlURL #2.html\##3}{\usebox{\logoagdabox}}}}
\newcommand{\agdarefcustom}[4]{%
  \coqrefcustom{#1}{#2}{#3}{\href{\onlineHtmlURL #2.html\##3}{\usebox{\logoagdabox}}}%
  {#4%
  }}

\usetikzlibrary{cd, calc, backgrounds, arrows.meta}

\usepackage{microtype}

\usepackage{ifthen}

\usepackage{environ}
\usepackage{newfile}

\newboolean{proofsinappendix}
\setboolean{proofsinappendix}{true}

\makeatletter
\ifthenelse{\boolean{proofsinappendix}}{
  \newoutputstream{proofstream}
  \openoutputfile{\jobname-proofs.out}{proofstream}
}{
}

\newcommand{\proofappendixbegin}[2]{%
  \phantomsection%
  \subsection*{\textbf{#1~\autoref{#2}}}%
  \addcontentsline{toc}{subsection}{#1~\autoref{#2}}%
  \label{#2:proof}%
  \def\proofappendix@qedsymbolmissing{\qed}
}
\newcommand{\proofappendixend}{%
  \proofappendix@qedsymbolmissing%
}
\let\oldqedhere\qedhere
\def\qedhere{\global\def\proofappendix@qedsymbolmissing{}\oldqedhere}

\newenvironment{proofhere}[2][Proof of]{%
  \proofappendixbegin{#1}{#2}%
}{%
  \proofappendixend%
  \par%
}

\ifthenelse{\boolean{proofsinappendix}}{%
  \NewEnviron{proofappendix}[2][Proof of]{%
    \addtostream{proofstream}{\noexpand\proofappendixbegin{#1}{#2}}%
    \addtostream{proofstream}{\noexpand\takeout{}}
    \expandafter\addtostream{proofstream}{\expandonce{\BODY}}%
    \addtostream{proofstream}{\noexpand\proofappendixend}%
    \addtostream{proofstream}{}%
  }%
}{%
  \newenvironment{proofappendix}[2][Proof of]{%
    \begin{proofhere}[#1]{#2}
  }{%
    \end{proofhere}
  }
}
\makeatother

\usepackage{newunicodechar}
\newunicodechar{⇒}{\ensuremath{\Rightarrow}}
\newunicodechar{ℰ}{\ensuremath{\mathscr{E}}}
\newunicodechar{₀}{\ensuremath{{}_{0}}}
\newunicodechar{₁}{\ensuremath{{}_{1}}}
\newunicodechar{₂}{\ensuremath{{}_{2}}}
\newunicodechar{λ}{\ensuremath{\lambda}}
\newunicodechar{≺}{\ensuremath{<}}
\newunicodechar{⊥}{\ensuremath{\bot}}
\newunicodechar{ε}{\ensuremath{\varepsilon}}
\newunicodechar{↓}{\ensuremath{\downarrow}}
\newunicodechar{⇃}{}
\newunicodechar{⇂}{}
\newunicodechar{↿}{}
\newunicodechar{⊒}{\ensuremath{\sge}}
\newunicodechar{♯}{\ensuremath{_{\#}}}
\newunicodechar{⊑}{\ensuremath{\sle}}
\newunicodechar{ₗ}{\ensuremath{_l}}
\newunicodechar{ᵣ}{\ensuremath{_r}}
\newunicodechar{⁻}{\ensuremath{^-}}
\newunicodechar{¹}{\ensuremath{^1}}
\newunicodechar{∷}{\ensuremath{\mathbin{::}}}
\newunicodechar{┌}{\ensuremath{\lceil}}
\newunicodechar{┐}{\ensuremath{\rceil}}
\newunicodechar{₁}{\ensuremath{_1}}
\newunicodechar{₂}{\ensuremath{_2}}
\newunicodechar{→}{\ensuremath{\to}}
\newunicodechar{⸴}{\ensuremath{,}}
\newunicodechar{ℬ}{\ensuremath{\mathcal{B}}}
\newunicodechar{⊔}{\ensuremath{\sqcup}}
\newunicodechar{⊓}{\ensuremath{\sqcap}}
\newunicodechar{⊎}{\ensuremath{\uplus}}
\newunicodechar{₀}{}

\newcommand{\Set}{\ensuremath{\mathsf{Set}}}

\newcommand{\Id}{\ensuremath{\mathsf{Id}}}
\newcommand{\inl}{\ensuremath{\mathsf{inl}}}
\newcommand{\inr}{\ensuremath{\mathsf{inr}}}
\newcommand{\alg}{\ensuremath{\mathsf{alg}}}
\newcommand{\coalg}{\ensuremath{\mathsf{coalg}}}
\newcommand{\prl}{\ensuremath{\mathsf{pr}_1}}
\newcommand{\prr}{\ensuremath{\mathsf{pr}_2}}
\newcommand{\id}{\ensuremath{\mathsf{id}}}

\newcommand{\pr}{\ensuremath{\mathsf{pr}}}
\newcommand{\C}{{\ensuremath{\mathcal{C}}}}
\newcommand{\D}{{\ensuremath{\mathcal{D}}}}

\newcommand{\floor}[1]{\ensuremath{\lfloor #1 \rfloor}}
\newcommand{\fpair}[1]{\ensuremath{\ensuremath{\langle #1 \rangle}}}

\newcommand{\modop}{\ensuremath{\mathrel{\mathrm{mod}}}}
\newcommand{\Coalg}{\ensuremath{\mathsf{Coalg}}}

\newcommand{\derive}{\mathrel{\overset{+}{\Rightarrow}}}
\renewcommand{\epsilon}{\ensuremath{\varepsilon}}

\newcommand{\sort}{\mathsf{sort}}

\newcommand{\Nat}{\mathbb{N}}
\newcommand{\sle}{\sqsubseteq} %
\newcommand{\sgt}{\sqsupset} %
\newcommand{\sge}{\sqsupseteq}
\newcommand{\N}{\mathbb{N}}
\newcommand{\Z}{\mathbb{Z}}
\newcommand{\coprime}{\mathrel{\bot}}
\newcommand{\B}{\mathcal{B}}
\newcommand{\Pow}{\mathcal{P}}
\newcommand{\Powf}{\Pow_\mathsf{f}}
\newcommand{\Trees}{\mathsf{Trees}}

 \usepackage{xspace}

\renewcommand{\gcd}{\mathsf{gcd}}
\newcommand{\NgN}{{\mathord{\Nat\times \Nat}}}
\newcommand{\CYK}{\textsf{CYK}\xspace}
\newcommand{\eea}{\mathsf{eea}}
\newcommand{\Int}{\mathbb{Z}}

\newcommand{\ol}{\overline}

\newcommand{\maxsteps}{\mathsf{maxsteps}}

\newcommand{\xto}{\xrightarrow}
\newcommand{\inj}{\mathsf{in}}
\newcommand{\seq}{\subseteq}
\newcommand{\takeout}[1]{}

\DeclareMathSymbol{\mathinvertedexclamationmark}{\mathclose}{operators}{'074}
\DeclareMathSymbol{\mathexclamationmark}{\mathclose}{operators}{'041}

\makeatletter
\newcommand{\raisedmathinvertedexclamationmark}{%
  \mathclose{\mathpalette\raised@mathinvertedexclamationmark\relax}%
}
\newcommand{\raised@mathinvertedexclamationmark}[2]{%
  \raisebox{\depth}{$\m@th#1\mathinvertedexclamationmark$}%
}
\begingroup\lccode`~=`! \lowercase{\endgroup
  \def~}{\@ifnextchar`{\raisedmathinvertedexclamationmark\@gobble}{\mathexclamationmark}}
\mathcode`!="8000
\makeatother

\newcommand{\hookto}{\ensuremath{\hookrightarrow}}

\newcommand{\monoto}{\ensuremath{\rightarrowtail}}

\newcommand{\set}[2][]{%
  \ifthenelse{\equal{#2}{}}{%
    \ensuremath{{#1\emptyset}}%
  }{%
    \ensuremath{{#1\{#2#1\}}}%
  }%
}

\newsavebox{\mypullbackcorner}%
\sbox{\mypullbackcorner}{%
\begin{tikzpicture}
    \draw[-] (0,0) -- (.5em,.5em) -- (0,1em);
\end{tikzpicture}%
}
\newcommand{\pullbackangle}[2][]{\arrow[phantom,to path={
                     -- ($ (\tikztostart)!1cm!#2:([xshift=8cm]\tikztostart) $)
                        node[anchor=west,pos=0.0,rotate=#2,
                        inner xsep = 0]
                        {\begin{tikzpicture}[minimum
                        height=1mm,baseline=0,#1]
    \draw[-] (0,0) -- (.5em,.5em) -- (0,1em);
                        \end{tikzpicture}}}]{}}

\tikzstyle{shiftarr}=[
        rounded corners,%
        to path={--([#1]\tikztostart.center)
                     -- ([#1]\tikztotarget.center) \tikztonodes
                     -- (\tikztotarget)},
]

\tikzset{
  loop at/.style={
    loop,
    out=#1-30,
    in=#1+30,
    looseness=6,
    every node/.append style={
      anchor=#1-180,
    },
  },
}

\newcommand{\descto}[3][]{\arrow[phantom]{#2}[font=\footnotesize,#1]{\text{\begin{tabular}{c}{}#3\end{tabular}}}}

\usepackage{fixme}
\fxsetup{silent,noinline,nomargin,footnote}
\fxusetheme{color}
\FXRegisterAuthor{ca}{aca}{CA}
\FXRegisterAuthor{hu}{ahu}{HU}
\FXRegisterAuthor{tw}{atw}{TW}

\RequirePackage[
  datamodel=acmdatamodel,
  style=acmnumeric,
  natbib
  ]{biblatex}

\newcommand{\resetCurThmBraces}{%
  \gdef\curThmBraceOpen{(}%
  \gdef\curThmBraceClose{)}}
\resetCurThmBraces
\newcommand{\removeThmBraces}{%
  \gdef\curThmBraceOpen{}%
  \gdef\curThmBraceClose{}}

\makeatletter
\patchcmd{\thmhead@plain}{(#3)}{\curThmBraceOpen #3\curThmBraceClose }{}{}%
\patchcmd{\thmhead@acmplain}{(#3)}{\curThmBraceOpen #3\curThmBraceClose }{}{}%
\patchcmd{\thmhead@acmdefinition}{(#3)}{\curThmBraceOpen #3\curThmBraceClose }{}{}%
\makeatother

\newenvironment{envcite}{\removeThmBraces}{\resetCurThmBraces}

\AtEndPreamble{
  \theoremstyle{acmdefinition}
  \newtheorem{assumption}[theorem]{Assumption}
  \newtheorem{convention}[theorem]{Convention}
  \newtheorem{observation}[theorem]{Observation}
  \newtheorem{remark}[theorem]{Remark}
  \newtheorem{notation}[theorem]{Notation}
  \renewcommand{\theHexample}{\thesection.\the\value{example}}
}

\title{Intrinsically Correct Algorithms and Recursive Coalgebras}

\author{Cass Alexandru}
\email{c.alexandru@cs.rptu.de}
\orcid{0009-0008-5210-2873}

\affiliation{%
  \institution{RPTU Kaiserslautern-Landau}
  \city{Kaiserslautern}
  \country{Germany}
}
\affiliation{%
  \institution{Radboud University Nijmegen}
  \city{Nijmegen}
  \country{The Netherlands}
}

\author{Henning Urbat}
\orcid{0000-0002-3265-7168}
\email{henning.urbat@fau.de}
\affiliation{%
  \institution{Friedrich-Alexander-Universität Erlangen-Nürnberg}
  \city{Erlangen}
  \country{Germany}
}

\author{Thorsten Wißmann}
\email{thorsten.wissmann@fau.de}
\orcid{0000-0001-8993-6486}
\affiliation{%
  \institution{Friedrich-Alexander-Universität Erlangen-Nürnberg}
  \city{Erlangen}
  \country{Germany}
}

\begin{abstract}
Recursive coalgebras provide an elegant categorical tool for modelling recursive algorithms and analysing their termination and correctness. By considering coalgebras over categories of suitably indexed families, the correctness of the corresponding algorithms follows \emph{intrinsically} just from the type of the computed maps. However, proving recursivity of the underlying coalgebras is non-trivial, and proofs are typically ad hoc. This layer of complexity impedes the formalization of coalgebraically defined recursive algorithms in proof assistants. We introduce a framework for constructing coalgebras which are \emph{intrinsically} recursive in the sense that the type of the coalgebra guarantees recursivity from the outset. Our approach is based on the novel concept of a \emph{well-founded functor} on a category of families indexed by a well-founded relation. We show as our main result that every coalgebra for a well-founded functor is recursive, and demonstrate that well-known techniques for proving recursivity and termination such as ranking functions are subsumed by this abstract setup. We present a number of case studies, including Quicksort, the Euclidian algorithm, and CYK parsing. Both the main theoretical result and selected case studies have been formalized in Cubical Agda.
\end{abstract}

\ifthenelse{\boolean{diffmode}}{%
\RequirePackage[normalem]{ulem} %
\RequirePackage{color}\definecolor{RED}{rgb}{1,0,0}\definecolor{BLUE}{rgb}{0,0,1} %
\RequirePackage{listings} %
\RequirePackage{color} %
\lstdefinelanguage{DIFcode}{ %
  moredelim=[il][\color{red}\sout]{\%DIF\ <\ }, %
  moredelim=[il][\color{blue}\uwave]{\%DIF\ >\ } %
} %
\lstdefinestyle{DIFverbatimstyle}{ %
	language=DIFcode, %
	basicstyle=\ttfamily, %
	columns=fullflexible, %
	keepspaces=true %
} %
\lstnewenvironment{DIFverbatim}{\lstset{style=DIFverbatimstyle}}{} %
\lstnewenvironment{DIFverbatim*}{\lstset{style=DIFverbatimstyle,showspaces=true}}{} %
\lstset{extendedchars=\true,inputencoding=utf8}
}

\makeatletter
\ifthenelse{\boolean{showmaintext}}{ }{
  \input{goodstein-pldi-src-main.aux}
}
\ifthenelse{\boolean{showappendix}}{ }{
  \input{goodstein-pldi-src-appendix.aux}
}
\makeatother

\setcopyright{cc}
\setcctype{by}
\acmDOI{10.1145/3808309}
\acmYear{2026}
\acmJournal{PACMPL}
\acmVolume{10}
\acmNumber{PLDI}
\acmArticle{231}
\acmMonth{6}
\acmSubmissionID{pldi26main-p434-p}
\received{2025-11-14}
\received[accepted]{2026-04-03}

\begin{document}



\ifthenelse{\boolean{showmaintext}}{

\maketitle

\section{Introduction}

Recursion is a powerful and widely used design principle for algorithms. However, as soon as an algorithm is not \emph{structurally recursive} on its input, such as is the case for most divide-and-conquer algorithms, syntactic termination checkers as used in total functional programming languages/proof assistants will reject a recursive definition.\twnote{der satz ist noch etwas lang} Instead, \emph{well-founded recursion} must be used (for a comparison of various ways to achieve this in the Rocq proof assistant, see~\cite{leroyWellfoundedRecursionDone2024}). One disadvantage of this approach is that it negates the possibility of disentangling the recursive steps from the recursive definition per se, following the paradigm of \emph{structured recursion}~\cite{hinzeConjugateHylomorphismsMother2015,CaprettaUV06,birdAlgebraProgramming1997,meijerFunctionalProgrammingBananas1991}. To bring structured recursion to the world of total functional programs,
a need arises for a \emph{general}, \emph{abstract}, \emph{compositional} methodology that supports this principle.

To this end, in this paper we develop novel sufficient criteria, amenable to formalization and formalized in a type-theoretical proof assistant, for defining recursive algorithms using \emph{recursive coalgebras}~\cite{osiusCategoricalSetTheory1974,amm25,taylor99,CaprettaUV06}. The latter are the conceptual foundation of structured recursion. The key idea of this abstract approach to programs is to decompose the recursive computation of a function $h\colon C\to A$ as shown in diagram \eqref{eq:rec-coalgebra}  into (1) a map $c\colon C\to FC$ (a \emph{coalgebra} for a suitable functor $F$) that splits an input from $C$ into \enquote{smaller} parts, (2) the recursive computation of $h$ on those parts, and (3) a map $a\colon FA\to A$ (an \emph{algebra} for $F$) that combines the recursively computed values back into a value of the target set $A$. The choice of the functor $F$ determines the type of data occurring in the three steps, as well as the structure of the call tree.
For example, the Quicksort algorithm (cf.\ \autoref{sec:quicksort}), which computes the map $\sort\colon Z^* \to Z^*$ sending a list of elements of some linearly ordered set \((Z, \sle)\) to its sorted permutation, corresponds to the set functor $FX=1+X\times Z\times X$ and the decomposition \eqref{eq:sort-qs}.

The coalgebra $c$ describes the choice of the pivot, say the head of the input list, and the splitting of the list into two sublists (e.g.\ $c([7,5,9,8,4])=([5,4],7,[9,8])$), and the algebra $a$ combines the recursively sorted sublists and the pivot into a single sorted list (e.g.\ $a([4,5],7,[8,9])=[4,5,7,8,9]$).

\noindent 
\begin{minipage}{.3\textwidth}
\begin{equation}\label{eq:rec-coalgebra}
\begin{tikzcd}[column sep=30, row sep=20]
C \ar[dashed]{r}{h} \ar{d}[swap]{c} &  A \\
FC \ar{r}{Fh} & FA \ar{u}[swap]{a} 
\end{tikzcd}
\end{equation}
\end{minipage}
\begin{minipage}{.01\textwidth}
  ~
\end{minipage}
\begin{minipage}{.63\textwidth}
\begin{equation}\label{eq:sort-qs}
  \begin{tikzcd}[column sep=70]
    Z^* \arrow{d}[swap]{c}  \arrow[dashed]{r}{\sort}  &  Z^*
   \\
   1 + Z^* \times Z\times Z^*
    \arrow{r}{\id+ \sort \times \id\times \sort} & 
    1 + Z^* \times Z \times Z^*
    \arrow{u}[swap]{a}
  \end{tikzcd}
\end{equation}
\end{minipage}

\smallskip\noindent Many recursive algorithms besides Quicksort can be presented diagrammatically in this style~\cite{CaprettaUV06,jeanninWellfoundedCoalgebrasRevisited2017}. The coalgebra $c$ underlying a recursive algorithm is called \emph{recursive} if, for each algebra $a$, there is a \emph{unique} morphism $h$ making diagram \eqref{eq:rec-coalgebra} commute, i.e.\ \(\forall (a \colon FA \to A).\, \exists ! (h \colon C \to A).\, h = a \cdot Fh \cdot c\). For instance, the Quicksort coalgebra $c\colon Z^*\to 1 + Z^* \times Z \times Z^*$ is recursive. We provide a proof of this using existing techniques (\autoref{prop:sort-recursive}), as well as using our novel criterion (\autoref{ex:quicksort-well-founded}).
Recursivity of a coalgebra $c$ guarantees that every recursive algorithm whose recursive branching is specified by $c$ terminates on every input and thus computes a total map, no matter what exactly the algebraic part $a$ of the algorithm does.

The (co)algebraic analysis of Quicksort can be augmented to also prove the \emph{correctness} of the algorithm. Following \citet{AlexandruCRW25}, the key observation is that diagram \eqref{eq:sort-qs} lifts from the category $\Set$ of sets to the category $\Set/\B Z$ of $\B Z$-indexed families of sets, where $\B Z$ is the set of finite multisets (bags) of integers. This entails that the map $\sort\colon Z^*\to Z^*$ computed by Quicksort preserves the multiset of elements of the input list $w$, so $\sort(w)$ is a permutation of~$w$.
The multiset index also allows expressing that the partition of elements around the pivot is indeed \emph{ordered}. This  means that the final list, obtained by collapsing a call tree equivalent to a binary search tree, is likewise ordered.
This is an instance of \emph{intrinsic} correctness: the correctness of the map $\sort$ follows immediately from its type, i.e.\ the fact that it is a morphism in $\Set/\B Z$ (cf.\ \autoref{sec:quicksort-correctness}). \citet{AlexandruCRW25} have shown that this principle extends from Quicksort to many other recursive sorting algorithms and provide an effective and uniform approach to verifying their correctness in Cubical Agda~\cite{vezzosiCubicalAgdaDependently2019}.

\subsection*{Contribution}
Our paper presents several new contributions to the theory of recursive coalgebras, directed towards enhancing their applicability to (fully formalized) correctness and termination proofs.

\paragraph*{Intrinsic Correctness via Recursive Coalgebras} We demonstrate that proofs of intrinsic correctness of recursive algorithms can be systematically based on recursive coalgebras lifted to categories of indexed families. Our approach is conceptually rather different from that of \citet{AlexandruCRW25} whose analysis of intrinsically correct sorting algorithms does not use the concept of recursivity, but rather exploits an initial algebra/final coalgebra coincidence in the category $\Set/\B Z$. The latter is specific to the list functor. The present approach is substantially more general and applies beyond the setting of sorting algorithms. We present several case studies for illustration, including Quicksort, the Euclidian algorithm, and the \CYK parsing algorithm.

\paragraph*{Well-Founded Functors} The key difficulty in working with recursive coalgebras lies in actually proving their recursivity. This is usually done on a per case basis and requires potentially complex arguments. In fact, while a number of general sufficient criteria for recursivity are known~\cite{amm25,taylor99,CaprettaUV06}, most of them are hard to apply, or do not apply at all, to the kind of indexed coalgebras occurring in correctness proofs for recursive algorithms (cf.\ \autoref{sec:challenge}). The core contribution of our paper addresses this issue: We introduce the novel concept of a \emph{well-founded functor} on a category of families indexed by a well-founded relation, such as the relation $<$ on $\B Z$ ordering multisets by their cardinality. Our main result (\autoref{allCausalRecursive}) asserts that:
\[\text{\emph{Every} coalgebra for a well-founded functor is recursive.} \]
In other words, coalgebras for well-founded functors are \emph{intrinsically} recursive. We demonstrate in our case studies that the functors underlying many recursive algorithms are well-founded, which usually follows directly from their definition. Thus the above result greatly simplifies recursivity proofs for the corresponding coalgebras, and moreover bases those proofs on a common foundation.

\paragraph*{Coalgebraic Termination Analysis}
On top of supporting proofs of intrinsic correctness, our theory of well-founded functors gives rise to a general coalgebraic account of proof methods for program termination. Terminating recursive programs admit a neat categorical abstraction in the form of \emph{well-founded coalgebras}, which are closely related to recursive coalgebras. We demonstrate that termination analysis techniques based on ranking functions~\cite{cookPodelskiRybalchenko}, including advanced techniques such as disjunctive well-foundedness~\cite{pr04}, emerge at the level of well-founded coalgebras in the category of sets, and in part even extend to abstract categories. 

\paragraph*{Formalization} Our intrinsic recursivity theorem for well-founded functors and a selection of case studies have been implemented in the proof assistant Cubical Agda. Besides providing a formal verification of the corresponding results, the implementation highlights how the coalgebraic framework developed in our paper yields a convenient and systematic foundation for the design of formalized correctness proofs.
Our formalization is available on:
\begin{center}
\nicehref{\onlineHtmlURL index.html}{\onlineHtmlURL index.html}
\end{center}
Individual formalized notions and results in the paper are marked by an icon
\href{\onlineHtmlURL index.html}{\usebox{\logoagdabox}} that is a link to the
respective file and identifier in the above linked web page. Additionally,
there is an explicit index of formalized results in \autoref{agdarefsection}\footnote{Available in the supplementary material to this paper.}.
We have focused the formalization of results to the ones relevant for the development of intrinsically correct algorithms.
In \autoref{sec:well-founded-functors} those are the results related to \emph{recursivity} of coalgebras.
The corresponding results on \emph{well-founded coalgebras} are not formalized, because they are only used for set-theoretic arguments on paper and not used in the implemented algorithms.
In \autoref{sec:well-founded-coreflection} the formalization is of the results specialized to the category of sets as modeled by Agda types -- the rationale for this is laid out in \autoref{sec:library-design}.
Of the case studies, the algorithmic ones have been machine-checked. These are the running Quicksort example as well as the GCD and CYK in \autoref{sec:case-studies}.

\section{Recursive Algorithms as Recursive Coalgebras}\label{sec:prelim}

In this section we motivate the notion of recursive coalgebra for modeling intrinsically correct algorithms.
In \ref{sec:quicksort} we give some background on recursive coalgebras and how they serve as an abstract model for recursive algorithms~\cite{taylor99,CaprettaUV06,jeanninWellfoundedCoalgebrasRevisited2017}.
In \ref{sec:well-found-coalg}, we recapitulate standard techniques \cite{amm25,taylor99} for proving recursivity of coalgebras.
In \ref{sec:quicksort-correctness}, we recapitulate the notion of intrinsic correctness.
We conclude in \ref{sec:challenge} with the observation that the indexed setting of intrinsically correct algorithms simultaneously provides a challenge to existing techniques and the motive for our novel technique for proving recursivity.
Throughout this section, we use QuickSort as a running example.

\subsection{Categorical Preliminaries}\label{sec:prelim-cat}
Throughout this paper we shall use some (very elementary) category
theory~\cite{awodey10,maclane98}; in particular, readers should be familiar
with functors, natural transformations, basic universal constructions such as
(co)products and pullbacks, and adjunctions. For the purposes of the present
paper, we focus on the following categories and functors on them:
\begin{itemize}
\item The category $\Set$ which consists of sets $X,Y,Z,\ldots $ as objects and
maps $f\colon X\to Y$ as morphisms. A functor $F\colon \Set\to\Set$ consists of
a construction that sends a set to another set, e.g. $X\mapsto X^2$ or $X\mapsto \N + X$, with the
property that $F$ is compatible with maps: for every map $f\colon X\to Y$ we
obtain a map $Ff\colon FX\to FY$ (preserving identities and function
composition).

\item For a set $I$, we let $\C^I$ denote the $I$-fold power of the category $\C$. Its objects are $I$-indexed families $X=(X_i)_{i\in I}$ of objects of $\C$, and a morphism $f\colon X\to Y$ is an $I$-indexed family of morphisms $(f_i\colon X_i\to Y_i)_{i\in I}$. Composition and identities are defined componentwise. In the case $\C=\Set$, we have the equivalence of categories
\begin{equation}
  \Set^I \simeq \Set/I
  \label{eq:familySlice}
\end{equation}
between $\Set^I$ and the \emph{slice category} $\Set/I$. The objects of the latter are pairs $(X,r)$ of a set $X$ and a map $r\colon X\to I$, and a morphism from $(X,r)$ to $(Y,s)$ is a map $f\colon X\to Y$ such that $r=s\cdot f$. The equivalence \eqref{eq:familySlice} identifies a family $X\in \Set^I$ with the pair $(\coprod_i X_i,r)\in \Set/I$ where $r(\inj_i(x))=i$. Conversely, it identifies a pair $(X,r)\in \Set/I$ with the family $(r^{-1}(i))_{i\in I}$ in $\Set^I$.

\end{itemize}

We use the categorical counterparts for many set constructions to use them in both $\Set$ and $\Set^I$:
we write
$1$ for the terminal object (e.g.\ the set $1=\set{*}$ in $\Set$) and $X\times
Y$ for the binary product of two
objects $X$ and~$Y$ (cartesian product in \Set). Dually, we write $0$ for the initial object ($\emptyset$ in $\Set$), $!`\colon
0\to X$ for the unique morphism into an object $X$, and $X+Y$ for the coproduct
of $X$ and $Y$, with injections $\inl\colon X\to X+Y$ and $\inr\colon Y\to
X+Y$. Arbitrary set-indexed coproducts are denoted by $\coprod_{i\in I} X_i$,
with injections $\inj_j\colon X_j\to \coprod_{i\in I} X_i$ for $j\in I$ (disjoint union in $\Set$). In $\Set^I$, all the above constructions correspond to performing the respective $\Set$-constructions in every component $i\in I$.
Given
a family $f_i\colon X_i\to Y$ ($i\in I$) of morphisms, we write $[f_i]_{i\in
I}\colon \coprod_{i\in I} X_i\to Y$ for the unique morphism with
$f_j=[f_i]_{i\in I}\cdot \inj_j$ for all $j\in I$. In the category $\Set$ of
sets and functions, $0$ and $1$ are given by the empty set and the singleton
set, products by cartesian products, and coproducts by disjoint unions.
We call a function \(\overline{f}\colon A \to Y\) defined as \(f \colon X \to Y\) on a subset \(A \subseteq X\) of the original domain \(X\) a \emph{restriction} of \(f\). Dually, for \(B \subseteq Y\), if the image of \(f\) lies in \(B\), we say \(f\) \emph{corestricts} to \(B\).



\begin{convention}\label{conv:set-functors}
For set functors $F\colon \Set\to\Set$ we tacitly assume that $F$ preserves inclusion of sets, that is, $A\seq B$ implies $FA\seq FB$ for all $A,B\in\Set$. This simplifies notation, holds in all our examples, and is essentially without loss of generality: For every set functor $F$, there is a functor~$F'$ that preserves inclusion and is naturally isomorphic to $F$ on non-empty sets~\cite[Sec.~III.4]{adamek1990automata}.
\end{convention}

\subsection{Running Example: Quicksort}\label{sec:quicksort}
To motivate the abstract coalgebraic concepts developed in the sequel, we consider as a running example the Quicksort algorithm~\cite{hoare62}, a divide-and-conquer algorithm for comparison-based
sorting. We have sketched the core ideas in the introduction; here we give a more detailed account. In general, sorting lists
over a set $Z$ with respect to a linear order $\mathord{\sle} \subseteq
Z\times Z$ amounts to a map
\[
  \sort\colon Z^* \longrightarrow Z^*
\]
on the set $Z^*=\coprod_{n\in \Nat} Z^n$ of finite lists over $Z$, sending a list to its sorted permutation. (A list $w\in Z^n$ is \emph{sorted} if $w(0)\sle w(1)\sle \cdots \sle w(n-1)$.) Quicksort computes this map recursively as 
\[ \sort(\epsilon)=\epsilon\qquad\text{and}\qquad \sort(pw)= \sort(w_{\sle p})\, p\, \sort(w_{\sgt p}) \]
for the empty list $\epsilon$ and  $p\in Z$, $w\in Z^*$. Here $w_{\sle p}$ denotes the sublist of $w$ containing those entries~$x$ with
$x\sle p$; analogously $w_{\sgt p}$ contains those entries $x$ with $x\sgt p$. The above recursive definition can be expressed diagrammatically, as observed by~\citet{CaprettaUV06}. Indeed, it states precisely that the map $\sort$ makes the rectangle \eqref{eq:sort} below commute, with the maps $c$ and $a$ given by
\begin{align}
  &c\colon Z^* \to 1+Z^*\times Z\times Z^*,&&
  c(\epsilon) = \inl, &&
  c(pw) = \inr(w_{\sle p}, p, w_{\sgt p}), \label{eq:sort-coalgebra} \\
  &a\colon 1+Z^* \times Z\times Z^*
  \rightarrow Z^*,
  &&
  a(\inl) = \epsilon,
  &&
  a(\inr(u,p,v)) = u\, p\, v. \label{eq:sort-algebra}
\end{align}
The map $c$ specifies the choice of the pivot element of a non-empty list (here the head of the list) and the splitting into two smaller sublists, while the map $a$ describes how to combine two recursively sorted sublists and the pivot back into a single list. By taking the endofunctor $T$ on $\Set$ given by
\begin{equation}\label{eq:F-quicksort} TX=1+X\times Z\times X \qquad\text{and}\qquad  Tf=\id_1+f\times \id_Z \times f, \end{equation}
we can display the commutative diagram \eqref{eq:sort} more compactly as \eqref{eq:sort-functor}.

\noindent\begin{minipage}{.60\textwidth}
\begin{equation}\label{eq:sort}
  \begin{tikzcd}[column sep=60]
    Z^* \arrow{d}[swap]{c}  \arrow{r}{\sort}  &  Z^*
   \\
   1 + Z^* \times Z \times  Z^*
    \arrow{r}{\id+ \sort \times \id \times \sort} & 
    1 + Z^* \times Z \times Z^*
    \arrow{u}[swap]{a}
  \end{tikzcd}
\end{equation}
\end{minipage}
\begin{minipage}{.01\textwidth}
~
\end{minipage}
\begin{minipage}{.3\textwidth}
 \begin{equation}\label{eq:sort-functor}
  \begin{tikzcd}[column sep=20]
    Z^* \arrow{d}[swap]{c}  \arrow{r}{\sort}  &  Z^*
   \\
   TZ^*
    \arrow{r}{T\sort} &
    TZ^*
    \arrow{u}[swap]{a}
  \end{tikzcd}
\end{equation}
\end{minipage}

\paragraph{Diagrams}
Here, a diagram in the category $\Set$ consists of sets as vertices and maps as
edges between them. A diagram \emph{commutes}, when the composition of maps
along any path between the same start and end vertex yield the same composed
map. E.g.\ in \eqref{eq:sort-functor}, the commutativity boils down to the
equality $\sort = a\cdot Tsort \cdot c$.

\subsection{Recursive Coalgebras}
The diagrammatic presentation \eqref{eq:sort} of Quicksort suggests the following categorical abstraction:

\begin{definition}[Recursive Coalgebra, \agdaref{Cubical.Categories.Functor.RecursiveCoalgebra}{isRecursive}{}] Let $F$ be an endofunctor on a category $\C$.
\begin{enumerate}
\item An \emph{$F$-algebra} is a pair $(A,a)$ of an object $A$ of $\C$ (the \emph{carrier} of the algebra) and a morphism $a\colon FA\to A$ (its \emph{structure}).
\item Dually, an \emph{$F$-coalgebra} is a pair $(C,c)$ of an object $C$ of $\C$ (\emph{the carrier} of the coalgebra) and a morphism $c\colon C\to FC$ (its \emph{structure}).
\item A \emph{coalgebra-to-algebra} morphism from an $F$-coalgebra $(C,c)$ to an $F$-algebra $(A,a)$ is a morphism $h\colon C\to A$ such that the square below commutes, i.e.\ \(h = a \cdot Fh \cdot c\):
\[
\begin{tikzcd}
C \ar{r}{h} \ar{d}[swap]{c} &  A \\
FC \ar{r}{Fh} & FA \ar{u}[swap]{a} 
\end{tikzcd}
\] 
\item A coalgebra $(C,c)$ is \emph{recursive} if for {every} algebra $(A,a)$ there exists a unique coalgebra-to-algebra morphism from $(C,c)$ to $(A,a)$.
\end{enumerate}
\end{definition}
We think of a coalgebra-to-algebra morphism as a recursive algorithm computing the function $h\colon C\to A$: The coalgebra $c$ breaks an input $x$ from $C$ into smaller parts $x_1,\ldots,x_n$ whose values $h(x_1),\ldots, h(x_n)$ are then recursively computed, and the algebra $a$ combines those values into the target value $h(x)$ in $A$. Recursivity of  $c$ asserts that, no matter what $a$ does, the map $h$ is uniquely defined by the recursive procedure specified by $c$ and $a$. For example, this is the case in Quicksort:

\begin{proposition}\label{prop:sort-recursive}
The coalgebra $c\colon Z^*\to 1+Z^*\times Z \times Z^*$ defined by \eqref{eq:sort-coalgebra} is recursive.
\end{proposition}
It follows that the map $\sort$ is uniquely determined by its recursive specification given by the commutative diagram \eqref{eq:sort}. But how is \autoref{prop:sort-recursive}, and recursivity of coalgebras in general, actually proven? One natural approach is to analyse the \emph{well-foundedness} of the coalgebra $c$.


\subsection{Well-Founded Coalgebras}\label{sec:well-found-coalg}
The notion of \emph{well-founded coalgebra} captures at an abstract level the idea that the call tree specified by a coalgebra $c$ contains no infinite paths, so that every recursive algorithm based on $c$ terminates.

\begin{definition}[Well-Founded Coalgebra] \label{defWellFoundedCoalg}
Let $F\colon \C\to \C$ be an endofunctor.
\begin{enumerate}
\item A \emph{subcoalgebra} $m\colon (S,s)\monoto (C,c)$ of an $F$-coalgebra $(C,c)$ is a coalgebra $(S,s)$ together with a monomorphism $m\colon S\monoto C$ in $\C$ such that the square below commutes: 
\begin{equation}\label{eq:cartesian}
\begin{tikzcd}
S \ar{r}{s} \ar[tail]{d}[swap]{m}  & FS \ar{d}{Fm} \\
C \ar{r}{c} & FC 
\end{tikzcd}
\end{equation}
The subcoalgebra is \emph{cartesian} if the above square is a pullback.
\item 
 A coalgebra $(C,c)$ is \emph{well-founded} if it has no proper cartesian subcoalgebras: for every cartesian subcoalgebra $m\colon (S,s)\monoto (C,c)$, the monomorphism $m$ is an isomorphism.
\end{enumerate}
\end{definition}

\begin{remark}
If $F$ preserves monos and $(C,c)$ is a coalgebra, every monomorphism $m\colon S\monoto C$ carries at most one subcoalgebra $m\colon (S,s)\monoto (C,c)$, that is, $s$ is uniquely determined by $m$ and $c$. In the case where $F$ is a set functor (where mono preservation comes essentially for free by \autoref{conv:set-functors}), this means that every subset $S\seq C$ of a coalgebra carries at most one subcoalgebra.
\end{remark}

\begin{example}[Graphs]\label{ex:wfCoalg}
Let $\Pow\colon \Set\to \Set$ be the powerset functor. A coalgebra $(C,c)$ for $\Pow$ corresponds to a directed graph with nodes $C$ and edges given by $x\to y$ iff $y\in c(x)$. A subset $S\seq C$ carries a cartesian subcoalgebra iff, for all $x\in C$,
\begin{equation}\label{eq:cartesian-pow}
x\in S \iff \text{all successors of $x$ lie in $S$.}
\end{equation}
Indeed, the left-to-right implication says that the square \eqref{eq:cartesian} commutes ($S$ is a subcoalgebra), and the right-to-left implication says that it is a pullback. From this, it follows that
\[ \text{$(C,c)$ is well-founded} \iff \text{$(C,c)$ has no infinite paths}. \]
 For $(\Longrightarrow)$, suppose that $(C,c)$ is well-founded. Then the set $S\seq C$ of all nodes that lie on no infinite path is a cartesian subcoalgebra by \eqref{eq:cartesian-pow}, and so $S=C$. To prove $(\Longleftarrow)$, we argue contrapositively. Suppose that $(C,c)$ is not well-founded, and let $S\subsetneq C$ be a proper cartesian subcoalgebra. Pick $x_0\in C\setminus S$ arbitrarily. By \eqref{eq:cartesian-pow}, some successor $x_1$ of $x_0$ lies in $C\setminus S$. By \eqref{eq:cartesian-pow} again, some successor $x_2$ of $x_1$ lies in $C\setminus S$. Iterating this argument yields an infinite path $x_0\to x_1\to x_2\to \cdots$.
\end{example}

This example can be lifted from $\Pow$-coalgebras to $F$-coalgebras for a set functor $F\colon \Set\to \Set$. For every set $X$, there is a canonical map
\[ \tau_X\colon FX\to \Pow X,\qquad 
 \tau_X(t)= \bigcap \{ M\seq X \mid t\in FM  \}.  \]
(Recall that $FM\seq FX$ if $M\seq X$ by \autoref{conv:set-functors}.) If $F$ preserves wide intersections (i.e.\ pullbacks of arbitrary non-empty families of monos), then $\tau_X(t)$ is simply the least subset $M\seq X$ such that $t\in FM$, called the \emph{support} of $t$. The \emph{canonical graph} for an $F$-coalgebra $(C,c)$ is the $\Pow$-coalgebra
\[ C\xto{c} FC \xto{\tau_C} \Pow C. \]

\begin{envcite}
\begin{proposition}[{\cite[Rem.~6.3.4]{taylor99}}]\label{prop:wf-set}
Suppose that $F\colon \Set\to\Set$ preserves wide intersections. Then an $F$-coalgebra is well-founded iff its canonical graph is well-founded, i.e.\ has no infinite paths.
\end{proposition}
\end{envcite}

For functors on $\Set$ (or more generally $\Set^I$ for an index set $I$), well-foundedness and recursivity are connected by the following result, a special case of the \emph{recursion theorem}~\cite[Cor.~8.5.5, 8.6.9]{amm25}. Recall that a \emph{preimage} is a pullback of a monomorphism along any morphism. 

\begin{theorem}[Recursion Theorem for Indexed Sets]\label{thm:wf-vs-rec} Let $F\colon \Set^I\to \Set^I$ be a functor.
\begin{enumerate}
\item Every well-founded $F$-coalgebra is recursive.
\item If $F$ preserves preimages, every recursive $F$-coalgebra is well-founded.
\end{enumerate}
\end{theorem}
The recursion theorem can be generalized from $\Set^I$ to abstract categories with well-behaved monomorphisms; see~\citet{amm25} for more details.

\begin{example}[Graphs]
For $F=\Pow\colon \Set\to\Set$, the first part of the recursion theorem states the familiar principle of \emph{well-founded recursion}: for every well-founded graph $c\colon C\to \Pow C$ and every algebra $a\colon \Pow A \to A$, there is unique map $h\colon C\to A$ such that $h(x)=a(h[c(x)])$ for all $x\in C$. \end{example}

Thanks to the recursion theorem, we are now in the position to establish recursivity of the Quicksort coalgebra $c\colon Z^*\to 1+Z^*\times Z\times Z^*$ given by \eqref{eq:sort-coalgebra}:

\begin{proof}[Proof of \autoref{prop:sort-recursive}]
  Note first that the functor $FX=1+X\times Z\times X$ preserves wide intersections; categorically, this follows from the standard fact that $+$ and $\times$ commute with pullbacks in $\Set$. The canonical graph of $c$ has nodes $Z^*$ and edges $pw\to w_{\sle p}$ and $pw\to w_{\sgt p}$ for $p\in Z$ and $w\in Z^*$. The empty list $\epsilon$ has no outgoing edges. This graph is clearly well-founded: since for all edges $u\to v$ we have that the list $v$ is strictly shorter than $u$, there are no infinite paths. It follows that the coalgebra $c$ is well-founded by \autoref{prop:wf-set}, hence recursive by \autoref{thm:wf-vs-rec}.
\end{proof}

\subsection{Intrinsic Correctness of Quicksort}\label{sec:quicksort-correctness}
The coalgebraic analysis of Quicksort given so far is not entirely satisfactory. While \autoref{prop:sort-recursive} shows that the recursive procedure terminates and defines a unique map $\sort\colon Z^*\to Z^*$, the type of this map does not yet force its
correctness: it is a priori not clear that the list
$\sort(w)$ is actually sorted and contains the same elements as the input list $w\in Z^*$.
To capture the correctness of Quicksort in the coalgebraic framework, \citet{AlexandruCRW25} propose to model Quicksort not over $\Set$, but rather over the slice category $\Set/I$ (or equivalently the category of $\Set^I$ of $I$-indexed sets) for the index set $I$ of finite multisets (a.k.a.\ \emph{bags}) over $Z$:
\[
  I := \B Z = \set{\mu\colon Z\to \N\mid \mu(z) = 0\text{ for all but finitely many }z\in Z}.
\]
In the following, we write $z\in \mu$ if $\mu(z)\neq 0$. Moreover, we write $q\colon Z^*\to \B Z$ for the map that sends a list to the multiset of its elements, and we denote the set of all sorted lists over $Z$ by
\[ Z^*_\sle := \set{w\in Z^*\mid \text{$w$ sorted}} \qquad \text{with the inclusion map}\qquad  \iota\colon Z_\sle^{*} \hookto Z^{*}. \]   
The key to obtaining correctness of Quicksort is the following observation:
\begin{observation}[Intrinsic Correctness]\label{obsSortCorrect}
Every morphism
\[
  h\colon
  (Z^*,q)
  \longrightarrow
  (Z^*_\sle,q\cdot \iota)
  \quad\text{ in }\quad\Set/\B Z
\]
is a correct sorting function, that is, $h(w)$ is a sorted permutation of $w$ for each $w\in Z^*$. Indeed:
\begin{itemize}
\item $h(w)$ is a sorted list, because $h(w)\in Z^*_\sle$.
\item \(h\) is an index-preserving map. Thus, $h(w)$ maps to the same multiset of elements as $w$. Therefore, it is a permutation of $w$.
\end{itemize}
\end{observation}
\citet{AlexandruCRW25} construct such a map $h$ by interpreting $Z^*$ and $Z^*_\sle$ as an initial algebra and a final coalgebra in $\Set/\B Z$, respectively.
They do this for the specific example of sorting, and do not give a description of how their approach could generalize to other algorithms.
We follow a different, generalizable strategy. We begin by lifting the coalgebra $(Z^*,c)$ from $\Set$ to $\Set/\B Z$ to intrinsically prove its recursivity.
To this end, let us consider the functor
\begin{equation}\label{eq:barF-quicksort}\overline{T}\colon \Set/\B Z\to \Set/\B Z, \qquad \overline{T}(X,r)=(\overline{T}_1(X,r),\bar r),  \end{equation}
where $\overline{T}_1\colon \Set/\B Z\to \Set$ is the functor given by
\begin{equation}\label{eq:barF1-quicksort} \overline{T}_1(X,r)= 1+\{ (u,p,v)\in X\times Z \times X \mid (\forall z\in r(u).\, z\sle p) \wedge (\forall z\in r(v).\, z\sge p) \} \seq FX \end{equation}
and the indexing map $\bar r\colon \overline{T}_1(X,r)\to \B Z$ is defined by
\[ \bar r(\inl)=\emptyset \qquad \text{and}\qquad \bar r(\inr(u,p,v)) = 
    r(u)\uplus \set{p} \uplus r(v). \]
  Here $\emptyset$ is the empty multiset and $\uplus$ is the union of finite multisets (adding up multiplicities).
  The functor~$\overline{T}$ is a restriction of \(T\) that enforces that the partition around the pivot respects the ordering, and (co)algebras for it correspond to multiset-index-preserving maps.

Now observe that the $T$-coalgebra structure $c\colon Z^*\to 1 + Z^* \times Z \times Z^*$ \eqref{eq:sort-coalgebra} underlying Quicksort co-restricts to a map 
$c\colon Z^*\to \overline{T}_1(Z^*,q)$ and that this co-restricted map makes the triangle on the left below commute. In other words, $((Z^*,q),c)$ forms a $\overline{T}$-coalgebra. Similarly, the $T$-algebra structure $a\colon 1 + Z^* \times Z \times Z^*\to Z^*$ \eqref{eq:sort-algebra} (co-)restricts to a map $a\colon \overline{T}_1(Z_\sle,q\cdot \iota)\to Z_\sle$ making the triangle on the right commute; thus $((Z_{\sle}^*,q\cdot \iota),a)$ forms a $\overline{T}$-algebra.
  \[
    \begin{tikzcd}
      Z^*
      \arrow{rr}{c}
      \arrow{rd}[swap]{q}
      & & \overline{T}_1(Z^*,q)
      \arrow{dl}{\bar q}
      \\
      & \B Z
    \end{tikzcd}
    \qquad
    \begin{tikzcd}
      \overline{T}_1(Z^*_\sle,q\cdot \iota)
      \arrow{rd}[swap]{\overline{q\cdot \iota}}
      \arrow{rr}{a}
      & & Z^*_\sle
      \arrow{dl}{q\cdot \iota}
      \\
      & \B Z
    \end{tikzcd}
  \]
We can now improve \autoref{prop:sort-recursive} to the following indexed version:
 \begin{proposition}\label{prop:sort-recursive-indexed}
  The $\overline{T}$-coalgebra $((Z^*,q),c)$ is recursive.
 \end{proposition}
 Due to the additional indexing, the proof is fairly convoluted when approached from scratch. We defer the proof to \autoref{sec:well-founded-functors} (\autoref{ex:quicksort-well-founded}), where the above proposition is derived in a principled manner as an instance of a general result. For now, we can use the proposition to deduce:
 \begin{theorem}[Correctness of Quicksort, \agdaref{QuickSort.QuickSort}{quickSort}{}]\label{thm:correctness-quicksort} Quicksort correctly sorts every input list.
 \end{theorem}
\begin{proof}By recursivity of $((Z^*,q),c)$, there exists a unique coalgebra-to-algebra morphism:
  \[
    \begin{tikzcd}
      (Z^*,q)
      \arrow[dashed]{r}{\sort}
      \arrow{d}[swap]{c}
      & (Z^*_\sle,q\cdot \iota)
      \\
      \overline{T}(Z^*,q)
      \arrow{r}{\overline{T} \sort}
      & \overline{T}(Z^*_\sle,q\cdot \iota) \ar{u}[swap]{a}
    \end{tikzcd}
    \qquad \text{in $\Set/\B Z$}.
  \]
  By definition of $c$ and $a$, this is precisely the map recursively computed by Quicksort.  By \autoref{obsSortCorrect}, it thus follows that it is indeed a correct sorting function. 
  \end{proof}

 We explain the details of the formalization of this proof later in \autoref{sec:library-design}, once the required coalgebraic machinery is set up. Let us note that the above style of reasoning is a showcase of \emph{intrinsic correctness}: by working with a recursive coalgebra in a category of suitably indexed families, the type of the computed map alone guarantees that it computes the correct function.  Other sorting algorithms like Mergesort or Insertion Sort can be proven to be intrinsically correct in very similar spirit by adapting the underlying functor; for example, for Mergesort one takes the functor $FX=1+Z+X\times X$ that models splitting a list into its first and second half. We devise further examples of intrinsic correctness proofs in \autoref{sec:case-studies}.

\subsection{The Fundamental Challenge}\label{sec:challenge} At the heart of our coalgebraic account of Quicksort and the proof of its intrinsic correctness are the recursivity results for the coalgebra $(Z^*,c)$ in $\Set$ (\autoref{prop:sort-recursive}) and the lifted coalgebra $((Z^*,q),c)$ in $\Set/\B Z$ (\autoref{prop:sort-recursive-indexed}). Even in the basic $\Set$ case, the proof rests on the fairly intricate recursion theorem (\autoref{thm:wf-vs-rec}) and the non-trivial fact that well-foundedness of coalgebras reduces to well-foundedness of their canonical graphs (\autoref{prop:wf-set}). In the $\Set/\B Z$ case, the situation is further complicated since canonical graphs are no longer available in the indexed setting, so that the notion of well-foundedness becomes less intuitive and more difficult to work with. Our framework of \emph{well-founded functors} developed below will provide the means to design coalgebras that are \emph{intrinsically} (i.e.\ by construction) both recursive and well-founded. This approach reduces much of the complexity behind the coalgebraic modelling of recursive algorithms, and thereby simplifies and unifies the steps needed for their formal verification in proof assistants.

\section{Well-Founded Functors}\label{sec:well-founded-functors}
In this section we introduce the key concept of our paper, \emph{well-founded functors}. Such a functor lives on a category $\C^I$ of families whose index set $I$ is equipped with some well-founded relation $<$, and well-foundedness of the functor expresses that its $i$th output component depends only on the input components indexed by $j<i$. In the following, this idea is developed more formally.

\begin{assumption}
  \label{indexCat}
  Throughout this paper, we fix a category $\C$ with set-indexed coproducts; in particular, $\C$ has an initial object $0$.
  Moreover, we fix an index set $I$ equipped with a well-founded relation $\mathord{<} \subseteq I\times I$.
  Well-foundedness of $<$ means that the corresponding graph with nodes $I$ and edges $i\to j$, iff $i>j$, is well-founded, that is, there is no infinite descending chain $i_0 > i_1 > i_2 > \cdots$. 
\end{assumption}

\begin{example}[Quicksort]\label{ex:quicksort-instance}
Quicksort corresponds to the instance $\C=\Set$ and $I=\B Z$ with the well-founded relation
$\mu < \nu$ iff $|\mu|<|\nu|$. Here $|\mu|=\sum_{z\in Z} \mu(z)$ is the size of a finite multiset $\mu\in \B Z$.
\end{example}

\begin{remark}
  \begin{enumerate}
  \item We denote the well-founded relation on $I$ by $<$ because it is a strict partial order in all our applications. However, our theory does not require 
  any order-theoretic properties of $<$.
  \item We will study coalgebras in the category $\C^I$ of $I$-indexed families. Note that the relation $<$ on $I$ is not taken into account in the category $\C^I$ itself, but only
when considering functors on $\C^I$.
  \end{enumerate}
\end{remark}
\begin{notation}
  For $i\in I$, we denote by ${<}i := \set{j \in I\mid j < i}$ the set of indices strictly smaller than~$i$. We have the two projection functors
  \begin{alignat*}{30}
    &\pr_{<i}\colon \C^I\to \C^{<i},\qquad
    &&
    \pr_{<i} X = (X_j)_{j<i},\qquad
    &&
    \pr_{<i} f = (f_j)_{j<i},\\
    &\pr_i\colon \C^I\to \C, && \pr_i X=X_i, && \pr_i f = f_i.
  \end{alignat*}
\end{notation}

\begin{definition}[Well-Founded Functor]
  A functor $G\colon \C^I\to\C^I$ is \emph{well-founded}
  if for every $i\in I$, the functor $\pr_i\cdot G\colon \C^I\to \C$ factors
  through the projection $\pr_{<i}\colon \C^I\to \C^{<i}$, that is, there exists a functor $G_{<i}$ such that the diagram below commutes up to natural isomorphism:
  \begin{equation}
    \forall i\in I\colon\quad
    \begin{tikzcd}
      \C^I \ar[phantom]{dr}[description]{\cong} 
      \arrow{r}{G}
      \arrow{d}[swap]{\pr_{<i}}
      &\C^I
      \arrow{d}{\pr_i}
      \\
      \C^{<i}
      \arrow[dashed]{r}{\exists G_{<i}}
      &
      \C
    \end{tikzcd}
    \label{wffun}
  \end{equation}
\end{definition}

\begin{remark}\label{rem:well-founded-functor-char}
This definition captures precisely the above intuition that the $i$th output of the functor $G$ is fully determined by its inputs with indices $j<i$. In our applications, the functor $G_{<i}$ can always be chosen such that \eqref{wffun} commutes on the nose, not just up to isomorphism. Existence of such a $G_{<i}$ is equivalent to the elementary condition
\begin{equation}\label{eq:wf-cond}  \pr_{<i}f = \pr_{<i}g \implies (Gf)_i=(Gg)_i \end{equation}
for all pairs $f,g$ of morphisms in $\C^I$. This follows from the observation that the functor $\pr_{<i}$ is surjective on morphisms, i.e.\ each morphism of $\C^{<i}$ extends to one of $\C^I$. Note that by considering identity morphisms, \eqref{eq:wf-cond} also entails the corresponding condition for pairs $X,Y$ of objects in $\C^I$:
\begin{equation}\label{eq:wf-cond-obj}
  \pr_{<i} X = \pr_{<i} Y \implies (GX)_i = (GY)_i.
\end{equation} 
\end{remark}

Our core result on well-founded functors is that they guarantee recursivity and well-foundedness of all their coalgebras; in other words, coalgebras are \emph{intrinsically} recursive and well-founded.

\begin{theorem}[Intrinsic Recursivity / Well-Foundedness]\label{allCausalRecursive}
  Let $G\colon \C^I\to\C^I$ be well-founded.
  \begin{enumerate}
  \item\label{allCausalRecursive:1} Every $G$-coalgebra is recursive (\agdarefcustom{\autoref{allCausalRecursive}.\ref{allCausalRecursive:1}}{Paper}{F-wf⇒F-coalgs-recursive}{}).
  \item\label{allCausalRecursive:2} Every $G$-coalgebra is well-founded.
\end{enumerate}
\end{theorem}
Note that \ref{allCausalRecursive:2} implies \ref{allCausalRecursive:1} for $\C=\Set$ (\autoref{thm:wf-vs-rec}), but for general categories, well-foundedness and recursiveness are independent.

\begin{proof}
  \begin{enumerate}
  \item Let $c\colon C\to GC$ be a $G$-coalgebra and let $a\colon GA\to A$ be a $G$-algebra. We need to show that there exists a unique coalgebra-to-algebra morphism $h\colon (C,c)\to (A,a)$. We construct the components $h_i\colon C_i\to
  A_i$ by well-founded recursion. Let $i\in I$ and suppose that $h_j\colon C_j\to A_j$ has been defined for all $j<i$; thus we have a morphism  
    \[
    g_i \colon \pr_{<i} C\to \pr_{<i} A
    \text{ in }\C^{<i}
    \quad\text{ given by }\quad
    (g_i)_j := h_j\colon C_j\to A_j.
  \]
  By well-foundedness of $G$, there is a natural isomorphism
  \[
    \phi_i\colon
    \pr_i \cdot G
    \overset{\cong}{\longrightarrow}
    G_{<i}\cdot \pr_{<i}.
  \]
  This gives rise to a canonical definition of $h_i\colon C_i\to A_i$:
  \[
    \begin{tikzcd}[column sep=30]
      C_i
      \arrow{r}{c_i}
      \arrow[dashed]{d}[left]{h_i}[right]{:=}
      & (G C)_i
      \arrow{r}{(\phi_i)_C}
      & G_{<i} (\pr_{<i} C)
      \arrow{d}{G_{<i} g_i}
      \\
      A_i
      & (G A)_i
      \arrow{l}{a_i}
      & G_{<i} (\pr_{<i} A)
      \arrow{l}{(\phi_i^{-1})_A}
    \end{tikzcd}
    \qquad\text{ in } \C.
  \]
  We verify that the just defined morphism $h = (h_i)_{i\in I}\colon C\to A$
  in $\C^I$ is a coalgebra-to-algebra morphism, which means that the square below commutes:
  \[
    \begin{tikzcd}
      C
      \arrow{r}{c}
      \arrow{d}[left]{h}
      & G C
      \arrow{d}{G h}
      \\
      A
      & G A
      \arrow{l}{a}
    \end{tikzcd}
  \]
We check this componentwise for each $i\in I$. We have $h_j = (g_i)_j$ for every $j<i$ and so
  \(
    \pr_{<i} h = g_i
  \).
Then the desired equality $h_i=a_i\cdot (Gh)_i\cdot c_i$ follows from the commutative diagram below:
  \begin{equation}
    \begin{tikzcd}
      C_i
      \arrow{r}{c_i}
      \arrow{d}[left]{h_i}
      & (G C)_i
      \arrow{r}{(\phi_i)_C}
      \descto{d}{Def.}
      \arrow[shiftarr={yshift=6mm}]{rr}{\id}
      & G_{<i} (\pr_{<i} C)
      \arrow[shift right=2]{d}[swap]{G_{<i} g_i}{=}
      \arrow[shift left=2]{d}{G_{<i} (\pr_{<i} h)}
      \arrow{r}{(\phi_i^{-1})_C}
      \descto[xshift=5mm]{dr}{Nat.}
      &[8mm] (GC)_i
      \arrow{d}{(Gh)_i}
      \\
      A_i
      & (G A)_i
      \arrow{l}{a_i}
      & G_{<i} (\pr_{<i} A)
      \arrow{l}{(\phi_i^{-1})_A}
      & (GA)_i
      \arrow{l}{(\phi_i)_A}
      \arrow[shiftarr={yshift=-6mm}]{ll}{\id}
    \end{tikzcd}
    \qquad
    \text{ in } \C.
    \label{diag:def:hi}
  \end{equation}
  For uniqueness, suppose that $u\colon (C,c)\to (A,a)$ is a coalgebra-to-algebra morphism, that is, $u\colon C\to A$ is a morphism in $\C^I$ with $u=a\cdot G
  u\cdot c$. We prove $u_i = h_i$ for all $i\in I$ by well-founded induction.
Let $i\in I$, and suppose that $u_j = h_j$ for all
  $j<i$, that is, $\pr_{<i} u = \pr_{<i} h$. Then
\[
    (Gu)_i
    \overset{\text{Nat.}}{=} (\phi_i^{-1})_A \cdot G_{<i}(\pr_{<i} u)\cdot (\phi_i)_C 
    \overset{\phantom{\text{Nat.}}}{=} (\phi_i^{-1})_A \cdot G_{<i}(\pr_{<i} h)\cdot (\phi_i)_C 
    \overset{\text{Nat.}}{=} (Gh)_i,
\]
which is obtained by replacing $h$ with $u$ in the diagram in \eqref{diag:def:hi}.
Therefore
  \[
    u_i = a_i \cdot (G u)_i\cdot c_i
    = a_i \cdot (G h)_i\cdot c_i
    = h_i.
  \]
  \item  Let $c\colon C\to GC$ be a $G$-coalgebra and $m\colon (S,s)\monoto (C,c)$ a cartesian subcoalgebra.
  We show that $m$ is an isomorphism in $\C^I$ by proving that $m_i$ is an
  isomorphism in $\C$ for every $i\in I$. The proof is by well-founded induction.
  Let $i\in I$, and suppose that $m_j\colon S_j\to C_j$ is an isomorphism in~$\C$ for all $j<i$. This means that the morphism $\pr_{<i} m\colon \pr_{<i} S\to \pr_{<i} C$ is
  an isomorphism in $\C^{<i}$.
  Consider the naturality square of $\phi_i$ from \eqref{wffun} for the morphism $m$:
  \[
    \begin{tikzcd}
      (G S)_i
      \arrow{d}[swap]{(Gm)_i}
      \arrow{r}{(\phi_i)_S}[below]{\cong}
      \descto{dr}{Nat.}
      & G_{<i}(\pr_{<i} S)
      \arrow{d}{G_{<i}(\pr_{<i} m)}[swap]{\cong}
      \\
      (GC)_i
      & G_{<i}(\pr_{<i} C)
      \arrow{l}[swap]{(\phi_i^{-1})_C}[below]{\cong}
    \end{tikzcd}
    \quad
    \text{ in }
    \C.
  \]
  Since functors preserve isomorphisms, $G_{<i}(\pr_{<i} m)$ is an isomorphism,
  and so $(Gm)_i$ is an isomorphism, being the composite of three
  isomorphisms. Since the lower square in the diagram below is a pullback and the outside commutes, we get a morphism $u$ making the remaining parts commute:
  \[
    \begin{tikzcd}
      C_i
      \arrow[bend right=20]{ddr}[swap]{\id}
      \arrow{r}{c_i}
      \arrow[dashed]{dr}{u}
      &
      (GC)_i
      \arrow[bend left=20]{dr}{(Gm)_i^{-1}}
      \\
      &
      S_i
      \arrow[>->]{d}[swap]{m_i}
      \arrow{r}{s_i}
      \pullbackangle{-45}
      & (G S)_i
      \arrow{d}{(Gm)_i}
      \\
      & C_i
      \arrow{r}{c_i}
      & (GC)_i
    \end{tikzcd}
    \quad
    \text{ in }
    \C.
  \]
Thus the monomorphism $m_i$ is a split epimorphism ($m_i\cdot u=\id$), and so it is an isomorphism. \qedhere
\end{enumerate}
\end{proof}

In many applications, the above theorem is instantiated to endofunctors $G$ on the slice category $\Set/I$. By extension, we say that such a functor is \emph{well-founded} if the corresponding endofunctor on $\Set^I$ obtained by pre- and postcomposing with the equivalence $\Set^I\simeq \Set/I$ \eqref{eq:familySlice} is well-founded. Thanks to the following sufficient criterion, checking well-foundedness is usually easy:

\begin{proposition}[Well-Foundedness of Functors on $\Set/I$]\label{prop:lifted-functor-well-founded}
If $G\colon \Set/I\to \Set/I$ satisfies
\[ \forall (X,r)\in \Set/I.\, \forall i\in I.\, \bar r^{-1}(i) \seq G_1(r^{-1}({<}i),r_{{<}i}),\] 
where $G(X,r)=(G_1(X,r),\bar r)$ and $r_{{<}i}\colon r^{-1}({<}i)\to I$ is the restriction of $r$, then $G$ is well-founded. (Note that $G_1(r^{-1}({<}i),r_{{<}i})\seq G_1(X,r)$ by \autoref{conv:set-functors}.)
\end{proposition}

\begin{proofappendix}{prop:lifted-functor-well-founded}
We prove that $G$ satisfies the criterion \eqref{eq:wf-cond} of \autoref{rem:well-founded-functor-char}. Given a morphism $f\colon (X,r)\to (Y,s)$ of $\Set/I$, we write $(f_i\colon r^{-1}(i)\to s^{-1}(i))_{i\in I}$ for the corresponding morphism of $\Set^I$ under the equivalence $\Set^I\simeq \Set/I$; thus $f_i$ is the domain-codomain restriction of the map $f\colon X\to Y$ to the preimages $r^{-1}(i)\seq X$ and $s^{-1}(i)\seq Y$. Moreover, we write $r_i\colon r^{-1}(i)\to I$ and $s_i\colon s^{-1}(i)\to I$ for the restrictions of $r$ and $s$. Then for the morphism $Gf\colon G(X,r)\to G(Y,r)$ of $\Set/I$, whose underlying map is $G_1 f\colon G_1(X,r)\to G_1(Y,s)$, we have the following commutative diagram for each $i\in I$ (the unlabelled arrows are set inclusions):
\[
\begin{tikzcd}[column sep=50]
\bar r^{-1}(i) \ar{r}{(Gf)_i=(G_1f)_i} \ar[shiftarr={xshift=-60},tail]{ddd}  \ar[tail]{d} & \bar s^{-1}(i) \ar[shiftarr={xshift=60},tail]{ddd} \ar[tail]{d} \\
G_1(r^{-1}({<}i),r_{{<}i}) \ar[equals]{d} & G_1(s^{-1}({<}i),s_{{<}i}) \ar[equals]{d} \\
G_1(\coprod_{j<i} r^{-1}(j),[r_j]_{j<i}) \ar[tail]{d} \ar{r}{G_1(\coprod_{j<i} f_j)} & G_1(\coprod_{j<i} s^{-1}(j),[s_j]_{j<i}) \ar[tail]{d} \\
G_1(X,r) \ar{r}{Gf} & G_1(Y,s)  
\end{tikzcd}
\]
The outside commutes by definition of $(Gf)_i$, and the lower cell by definition of $f_j$. Therefore the upper cell commutes, which shows that $(Gf)_i$ is uniquely determined by the morphisms $f_j$ for $j<i$. This proves \eqref{eq:wf-cond}, so $\overline{T}$ is well-founded.
\end{proofappendix}

\begin{example}[Quicksort, \agdaref{QuickSort.QuickSort}{QuickSortFunctorWellfounded}{}]\label{ex:quicksort-well-founded}
The functor $\overline{T}$ \eqref{eq:barF-quicksort} on $\Set/\B Z$ used for modelling Quicksort is well-founded. To see this, we check the criterion of \autoref{prop:lifted-functor-well-founded}. We need to show:
\[\text{For $(X,r)\in \Set/\B Z$ and $\mu\in \B Z$, every element of the preimage $\bar r^{-1}(\mu)$ lies in $\overline{T}_1(r^{-1}({<}\mu),r_{<\mu})$}.\]
This is immediate from the definition of $\overline{T}_1$ \eqref{eq:barF1-quicksort}. Indeed,
if $\mu=\emptyset$, the only element sent by $\bar r$ to $\mu$ is $\inl$, which lies in $1$ and thus in $\overline{T}_1(r^{-1}({<}\mu),r_{<\mu})$. If $\mu\neq \emptyset$, then $\inr(u,p,v)$ is sent by $\bar r$ to $\mu$ iff $r(u)\uplus \{p\}\uplus r(v)=\mu$, which entails $r(u),r(v)<\mu$, and thus $\inr(u,p,v)$ lies in $\overline{T}_1(r^{-1}({<}\mu),r_{<\mu})$.

 We can now apply \autoref{allCausalRecursive} to get the missing proof of \autoref{prop:sort-recursive-indexed}: the indexed coalgebra $((Z^*,q),c)$ modelling Quicksort over $\Set/\B Z$ is recursive.
\end{example}

\section{Well-Founded Coreflection}\label{sec:well-founded-coreflection}
If a functor is not well-founded, it can always be transformed into a well-founded functor in a universal manner by forming its \emph{well-founded coreflection}.
This rather simple construction is instrumental for the standard technique we develop for proving well-foundedness of a functor.
\begin{definition}
  \label{def:trunc}
  For every $i\in I$, we define the inclusion functor $J_{<i}\colon \C^{<i}\to
  \C^I$ and the truncation functor $T_{<i}\colon \C^I\to \C^I$ as follows (recall that $\C$ has an initial object $0$ by \autoref{indexCat}):
  \[
    (J_{<i} Y)_j := \begin{cases}
      Y_j &\text{if }j < i,\\
      0&\text{otherwise},\\
    \end{cases}
    \qquad
    T_{<i} := J_{<i}\cdot \pr_{<i},
    \qquad
    (T_{<i} X)_j = \begin{cases}
      X_j &\text{if }j < i,\\
      0&\text{otherwise}.\\
    \end{cases}
  \]
\end{definition}
\begin{remark}
  \label{truncationAdjunction}
The functor $J_{<i}$ is left adjoint to $\pr_{<i}$ and so $\C^{<i}$ is a
\emph{coreflective} subcategory~\cite[18.2(1)]{joyofcats}. By this adjunction, the
truncation functor carries the structure of a comonad:
  \[
    \begin{tikzcd}[sep=15mm]
      \C^{<i}
      \arrow[bend left=20]{r}[alias=L]{J_{<i}}
      & \C^I
      \arrow[bend left=20]{l}[alias=R]{\pr_{<i}}
      \arrow[from=R,to=L,phantom]{}[sloped,description]{\vdash}
      \arrow[loop at=0,looseness=6]{}[right]{T_{<i}}
    \end{tikzcd}
  \]
The counit $t_i\colon T_{<i}\to \Id$ is the natural transformation with components
\[
  t_{i,X,j}\colon (T_{<i} X)_j \to X_j,
  \qquad
  t_{i,X,j} = \begin{cases}
    \id_{X_j}\colon X_j\to X_j &\text{if }j < i, \\
    !`\colon 0\to X_j &\text{otherwise}.
  \end{cases}
\]
\end{remark}

\begin{definition}[Well-Founded Coreflection]
  Let $G\colon \C^I\to \C^I$ be a functor. The \emph{well-founded coreflection} of $G$ is the functor
  \[
    G_\downarrow\colon \C^I \to \C^I,
    \qquad
    (G_\downarrow X)_i := (G T_{<i} X)_i.
  \]
The functors $G_\downarrow$ and $G$ are connected by the natural transformation $\epsilon_G\colon G_\downarrow \to G$ with components
  \[
    (\epsilon_{G,X})_i\colon (GT_{<i} X)_i \to (GX)_i,
    \qquad
    (\epsilon_{G,X})_i := (Gt_{i,X})_i.
  \]
\end{definition}
We refer to the functor $G_\downarrow$ as a \emph{coreflection}
because the next result shows that under a mild condition, the well-founded functors form a coreflective subcategory of the category of all endofunctors.
\begin{theorem}[Well-Founded Coreflection]\label{prop:well-founded-coflection} Let $G\colon\C^I\to\C^I$ be a functor.
 \begin{enumerate}
  \item\label{prop:well-founded-coreflection:1}  The functor $G_\downarrow\colon \C^I\to\C^I$
  is well-founded.
  \item\label{prop:well-founded-coreflection:2} The functor $G$ is well-founded iff $\epsilon_G\colon G_\downarrow\cong G$ is a natural isomorphism.
  \item\label{prop:well-founded-coreflection:3} If $\epsilon_G\colon G_\downarrow\to G$ is componentwise monic, then $\epsilon_G$ is co-universal.
 \end{enumerate}
\end{theorem}

\begin{proofappendix}{prop:well-founded-coflection}
\begin{enumerate} 
  \item For $i\in I$, the functor 
  \[
    G_{<i}\colon \C^{<i}\to \C,\qquad G_{<i} Y := (GJ_{<i} Y)_i,
  \]
   witnesses the natural isomorphism required in
  \eqref{wffun}, which is in fact an equality:
  \begin{align*}
    G_{<i} (\pr_{<i} X)
    = (G J_{<i} \pr_{<i} X)_i
    = (GT_{<i} X)_i
    = (G_\downarrow X)_i
    =
    \pr_i(G_\downarrow X).
  \end{align*}
  All these equational steps hold by definition (of $G_{<i}$, $T_{<i}$, $G_\downarrow$, $\pr_i$).
\item If $G_\downarrow\cong G$, then $G$ is well-founded because $G_\downarrow$ is well-founded by part \ref{prop:well-founded-coreflection:1}, and well-foundedness is clearly preserved by natural isomorphisms. Conversely, suppose that $G$ is well-founded. Then for each $i\in I$ there exists a functor $G_{<i}\colon \C^{<i}\to \C$ with a natural isomorphism
\[ \phi_{i,X}\colon (\pr_{i}\cdot G)X \xto{\cong} (G_{<i}\cdot \pr_{<i})X \qquad (X\in \C^I). \]    
Consider the diagram below:
\[
\begin{tikzcd}[column sep=50]
(GT_{<i}X)_i \ar[equals]{d} \ar{r}{(Gt_{i,x})_i} & (GX)_i \ar[equals]{d} \\
(pr_i\cdot G)T_{<i} X \ar{r}{(\pr_{i}\cdot G) t_{i,X}} \ar{d}[swap]{\phi_{i,T_{<i}X}}  & (\pr_i\cdot G) X \ar{d}{\phi_{i,X}} \\
(G_{<i}\cdot \pr_{<i})T_{<i} X \ar[equals]{d} \ar{r}{(G_{<i}\cdot \pr_{<i})t_{i,X}} & (G_{<i}\cdot \pr_{<i})X \ar[equals]{d} \\
(G_{<i}\cdot \pr_{<i}) X \ar{r}{\id} & (G_{<i}\cdot \pr_{<i}) X
\end{tikzcd}
\]
The upper cell commutes by definition of $\pr_i$, the middle cell by naturality of $\phi_i$, and the lower cell because $\pr_{<i}\cdot J_{<i}=\Id$ and $\pr_{<i}(t_{i,X})=\id$. It follows that the outside of the diagram commutes, so $(\phi_{i,T_{<i}X})^{-1}\cdot (\phi_{i,X})^{-1}$ is an inverse of $(Gt_{i,X})_i=(\epsilon_{G,X})_i$. Thus $\epsilon_G$ is an isomorphism.
\item Given $H$ and $\alpha$, we show how to construct the unique natural transformation $\ol{\alpha}\colon H\to G_\downarrow$ such that $\alpha=\epsilon_G\cdot \ol{\alpha}$, which means that for every $X\in \C$ and $i\in I$ the triangle below commutes:
\begin{equation}\label{eq:bar-alpha}
\begin{tikzcd}
& (GX)_i & \\
(HX)_i \ar{ur}{(\alpha_{X})_i} \ar[dashed]{rr}[swap]{(\ol{\alpha}_X)_i} & & (G_\downarrow X)_i  \ar{ul}[swap]{(\epsilon_{G,X})_i}
\end{tikzcd}
\end{equation}
Since $H$ is well-founded, we know from part \ref{prop:well-founded-coreflection:2} that $(\epsilon_{H,X})_i=(Ht_{i,X})_i\colon (H_\downarrow X)_i\to (HX)_i$ is an isomorphism. We define the desired natural transformation $\ol\alpha$ as follows:
\[ (\ol{\alpha}_X)_i \;\equiv\; (\, (HX)_i \xto{(Ht_{i,X})_i^{-1}} (H_\downarrow X)_i = (HT_{<i}X)_i \xto{(\alpha_{T_{<i}X})_i} (GT_{<i} X)_i = (G_\downarrow X)_i \,). \] 
Clearly $\ol\alpha$ is natural because both $t_{i,X}$ and $\alpha$ are natural. Moreover \eqref{eq:bar-alpha} commutes: composing $(\ol\alpha_X)_i$ with $(\epsilon_{G,X})_i = (Gt_{i,X})_i$ yields 
\[ (HX)_i \xto{(Ht_{i,X})_i^{-1}} (HT_{<i}X)_i \xto{(\alpha_{T_{<i}X})_i} (GT_{<i} X)_i \xto{(Gt_{i,X})_i} (GX)_i, \, \]
which is equal to $(\alpha_X)_i$ by naturality of $\alpha$. Uniqueness of $\ol{\alpha}$ follows from $\epsilon_G$ being componentwise monic by assumption.
\qedhere
\end{enumerate} 
\end{proofappendix}

\begin{remark}
Co-universality means that for every well-founded functor $H\colon \C^I\to \C^I$ and every natural transformation $\alpha\colon H\to G$, there exists a unique $\ol{\alpha}\colon H\to G_\downarrow$ such that the triangle
\[
\begin{tikzcd}[sep=1cm]
& G & \\
H \ar{ur}{\alpha} \ar[dashed]{rr}{\ol{\alpha}} & & G_\downarrow \ar{ul}[swap]{\epsilon_G}
\end{tikzcd}
\]
commutes. The condition that $\epsilon_G$ is componentwise monic guarantees uniqueness of $\ol{\alpha}$ and is fairly mild. For example, it holds whenever the initial object $0$ is \emph{simple} (i.e.\ every morphism $0\to X$ in $\C$ is monic) and $F$ preserves monomorphisms. Indeed, in that case $t_{i,X}\colon T_{<i}X \to X$ is monic because each component is either an identity morphism or has domain $0$, and so $(\epsilon_{G,X})_i=(Gt_{i,X})_i$ is monic.    
\end{remark}

\subsection{Agda Library Interface}
\label{sec:library-design}
\newcommand{\af}[1]{\AgdaFunction{#1}}
\newcommand{\aic}[1]{\AgdaInductiveConstructor{#1}}

In this section we explain how we mechanize above theorems as a library for writing intrinsically correct algorithms.
Such algorithms are defined as \(G\)-coalgebra-to-algebra morphisms for some well-founded functor \(G\).
The well-foundedness of \(G\) is verified via the well-founded coreflection $G_\downarrow$ using \autoref{prop:well-founded-coflection}\ref{prop:well-founded-coreflection:2}.
%


\paragraph{Mechanization of definitions}
When defining $G_\downarrow$, the case-distinction $j < i$ of $J_{<i}$ (\autoref{def:trunc}) requires the order \(<\) to be decidable.
In addition to this restriction, such a case-distinction is inconvenient to work with, because whenever
one has an element of \((J_{<i} X)_j\) at hand, one then needs to do another case-distinction for \(i < j\).
Thus, we work with a slightly different definition \(J'_{<i}\)
for the special case \(\C = \Set\), which we use in all the examples in \autoref{sec:case-studies}. Intuitively, we replace decidability of \(<\)
with a proof of \(j < i\) at the point where an element of \((J'_{<i} X)_j\) is introduced.
\begin{definition}[Inclusion Functor for \(\C\) specialized to \(\Set\), \agdaref{IntrinsicallyRecursiveCoalgs}{J<}{}]
  \label{def:trunc-set}
  For every $i\in I$, we define the inclusion functor $J'_{<i}\colon \Set^{<i}\to
  \Set^I$ as follows:
  \[
    (J'_{<i} X)_j := \{ x \mid j < i , x \in X_j \}
  \]
\end{definition}
The category \(\Set\) corresponds in our mechanization to the category \af{Type} of (Agda) types, which is a setting suitable to most formalized algorithms, including the ones in this paper.

\paragraph{Using the library to write recursive algorithms}
For a \(G\)-coalgebra \((C,c)\) and an algebra \((A,a)\),
our library provides an \emph{incremental} interface:
In its simplest form, it only takes an arbitrary $\Set^I$-morphism \(\epsilon'_i \colon (GX)_i \to (GT_{<i} X)_i\) and returns a $\Set^I$-morphism \(h \colon C \to A\) between the carriers (\agdarefcustom{Library stage 1}{IntrinsicallyRecursiveCoalgs}{c2a.c2a-morph-Definition.c2a-morph}{}).
The second stage requires that $\epsilon'$ is inverse to \(ε \colon G_↓ \to G\) and returns a proof that \(h\) is a coalgebra-to-algebra morphism from $(C,c)$ to $(A,a)$ (\agdarefcustom{Library stage 2}{IntrinsicallyRecursiveCoalgs}{c2a.c2a-morph-Definition.c2a.is-c2a-morph}{}).
This is a propositional equality witnessing that the algorithm fulfills the functional equation corresponding to its usual explicitly recursive formulation.
The third and final stage requires naturality of \(ε'\) and returns a proof that \(h\) is the \emph{unique} solution to this equation
(\agdarefcustom{Library stage 3}{IntrinsicallyRecursiveCoalgs}{c2a.c2a-morph-Definition.c2a.unique.unique}{}).

Let us now illustrate the library client code for the instance of Quicksort (\agdarefcustom{\autoref{sec:library-design}}{QuickSort.QuickSort}{quickSort}{intrinsically correct Quicksort as a coalgebra-to-algebra morphism}). The datatypes and some lemmas are taken from an existing formalization~\cite{alexandru_2024_14279034} of Quicksort in Cubical Agda. We use these to show that their functor is well-founded and thus fits into our framework.
For this application, as the correctness is already fully encoded at the level of the underlying map (cf.\ \autoref{sec:quicksort-correctness}), it suffices to define the $\epsilon'\colon G \to G_\downarrow$. The type constructor~\af{S} shown below corresponds to the functor $\overline{T}$ of \eqref{eq:barF-quicksort}. We define it along with a \emph{pattern synonym} for \aic{┌\_┐ } which allows us to introduce the implicit indices while still using infix notation: $u\AgdaOperator{\AgdaInductiveConstructor{\textasciicircum{}}}i_l$ is a list $u$ with elements $i_l$.
\begin{code}%
\>[0]\AgdaKeyword{data}\AgdaSpace{}%
\AgdaDatatype{S}\AgdaSpace{}%
\AgdaSymbol{(}\AgdaBound{X}\AgdaSpace{}%
\AgdaSymbol{:}\AgdaSpace{}%
\AgdaDatatype{ℬ}\AgdaSpace{}%
\AgdaBound{A}\AgdaSpace{}%
\AgdaSymbol{→}\AgdaSpace{}%
\AgdaPrimitive{Type}\AgdaSpace{}%
\AgdaGeneralizable{⇃}\AgdaSymbol{)}\AgdaSpace{}%
\AgdaSymbol{:}\AgdaSpace{}%
\AgdaDatatype{ℬ}\AgdaSpace{}%
\AgdaBound{A}\AgdaSpace{}%
\AgdaSymbol{→}\AgdaSpace{}%
\AgdaPrimitive{Type}\AgdaSpace{}%
\AgdaBound{⇃}\AgdaSpace{}%
\AgdaKeyword{where}\<%
\\
\>[0][@{}l@{\AgdaIndent{0}}]%
\>[2]\AgdaInductiveConstructor{leaf}\AgdaSpace{}%
\AgdaSymbol{:}\AgdaSpace{}%
\AgdaDatatype{S}\AgdaSpace{}%
\AgdaBound{X}\AgdaSpace{}%
\AgdaInductiveConstructor{[]}\<%
\\
\>[2]\AgdaOperator{\AgdaInductiveConstructor{\AgdaUnderscore{}┌\AgdaUnderscore{}┐\AgdaUnderscore{}}}\AgdaSpace{}%
\AgdaSymbol{:}%
\>[11]\AgdaSymbol{\{}\AgdaBound{iₗ}\AgdaSpace{}%
\AgdaBound{iᵣ}\AgdaSpace{}%
\AgdaSymbol{:}\AgdaSpace{}%
\AgdaDatatype{ℬ}\AgdaSpace{}%
\AgdaBound{A}\AgdaSymbol{\}}\AgdaSpace{}%
\AgdaSymbol{→}\AgdaSpace{}%
\AgdaSymbol{(}\AgdaBound{u}\AgdaSpace{}%
\AgdaSymbol{:}\AgdaSpace{}%
\AgdaBound{X}\AgdaSpace{}%
\AgdaBound{iₗ}\AgdaSymbol{)}\AgdaSpace{}%
\AgdaSymbol{→}\AgdaSpace{}%
\AgdaSymbol{(}\AgdaBound{p}\AgdaSpace{}%
\AgdaSymbol{:}\AgdaSpace{}%
\AgdaBound{A}\AgdaSymbol{)}\AgdaSpace{}%
\AgdaSymbol{→}\AgdaSpace{}%
\AgdaSymbol{(}\AgdaBound{v}\AgdaSpace{}%
\AgdaSymbol{:}\AgdaSpace{}%
\AgdaBound{X}\AgdaSpace{}%
\AgdaBound{iᵣ}\AgdaSymbol{)}\AgdaSpace{}%
\AgdaSymbol{→}\<%
\\
\>[2][@{}l@{\AgdaIndent{0}}]%
\>[5]\AgdaBound{p}\AgdaSpace{}%
\AgdaOperator{\AgdaFunction{⊒}}\AgdaSpace{}%
\AgdaBound{iₗ}\AgdaSpace{}%
\AgdaSymbol{→}\AgdaSpace{}%
\AgdaBound{p}\AgdaSpace{}%
\AgdaOperator{\AgdaFunction{⊑}}\AgdaSpace{}%
\AgdaBound{iᵣ}\AgdaSpace{}%
\AgdaSymbol{→}\AgdaSpace{}%
\AgdaDatatype{S}\AgdaSpace{}%
\AgdaBound{X}\AgdaSpace{}%
\AgdaSymbol{(}\AgdaBound{p}\AgdaSpace{}%
\AgdaOperator{\AgdaInductiveConstructor{∷}}\AgdaSpace{}%
\AgdaBound{iₗ}\AgdaSpace{}%
\AgdaOperator{\AgdaFunction{++}}\AgdaSpace{}%
\AgdaBound{iᵣ}\AgdaSymbol{)}\<%
\\
\>[0]\AgdaKeyword{pattern}\AgdaSpace{}%
\AgdaOperator{\AgdaInductiveConstructor{\AgdaUnderscore{}\textasciicircum{}\AgdaUnderscore{}┌\AgdaUnderscore{}┐\AgdaUnderscore{}\textasciicircum{}\AgdaUnderscore{}}}\AgdaSpace{}%
\AgdaBound{u}\AgdaSpace{}%
\AgdaBound{iₗ}\AgdaSpace{}%
\AgdaBound{p}\AgdaSpace{}%
\AgdaBound{v}\AgdaSpace{}%
\AgdaBound{iᵣ}\AgdaSpace{}%
\AgdaBound{pf₁}\AgdaSpace{}%
\AgdaBound{pf₂}\AgdaSpace{}%
\AgdaSymbol{=}\AgdaSpace{}%
\AgdaOperator{\AgdaInductiveConstructor{\AgdaUnderscore{}┌\AgdaUnderscore{}┐\AgdaUnderscore{}}}\AgdaSpace{}%
\AgdaSymbol{\{}\AgdaBound{iₗ}\AgdaSymbol{\}}\AgdaSpace{}%
\AgdaSymbol{\{}\AgdaBound{iᵣ}\AgdaSymbol{\}}\AgdaSpace{}%
\AgdaBound{u}\AgdaSpace{}%
\AgdaBound{p}\AgdaSpace{}%
\AgdaBound{v}\AgdaSpace{}%
\AgdaBound{pf₁}\AgdaSpace{}%
\AgdaBound{pf₂}\<%
\end{code}
\tikzset{
  agdamarkerframe/.style={
    inner xsep=1pt,
    inner ysep=2pt,
    outer ysep=-2pt,
    fill=black!10,
    rounded corners=2pt,
  },
  description node/.style={
    fill=black!90,
    rounded corners=2pt,
    text=white,
    font=\small,
    align=center,
  },
  description arrow/.style={
    line width=1pt,
    ->,
    shorten >=2pt, 
    >={Straight Barb[scale=0.6]},
    draw=black!90,
  },
}
\newcommand{\agdamark}[2]{\tikz[remember picture, baseline=(lasttikznode.base)]{\node[agdamarkerframe,alias=#1] (lasttikznode) {#2};}}
\begin{code}%
\>[0][@{}l@{\AgdaIndent{1}}]%
\>[2]\AgdaFunction{S-ε⁻¹}\AgdaSpace{}%
\AgdaSymbol{:}\AgdaSpace{}%
\AgdaSymbol{\{}\AgdaBound{X}\AgdaSpace{}%
\AgdaSymbol{:}\AgdaSpace{}%
\AgdaDatatype{ℬ}\AgdaSpace{}%
\AgdaBound{A}\AgdaSpace{}%
\AgdaSymbol{→}\AgdaSpace{}%
\AgdaPrimitive{Type}\AgdaSymbol{\}}\AgdaSpace{}%
\AgdaSymbol{→}\AgdaSpace{}%
\AgdaSymbol{(}\AgdaBound{i}\AgdaSpace{}%
\AgdaSymbol{:}\AgdaSpace{}%
\AgdaDatatype{ℬ}\AgdaSpace{}%
\AgdaBound{A}\AgdaSymbol{)}\AgdaSpace{}%
\AgdaSymbol{→}\AgdaSpace{}%
\AgdaDatatype{S}\AgdaSpace{}%
\AgdaBound{X}\AgdaSpace{}%
\AgdaBound{i}\AgdaSpace{}%
\AgdaSymbol{→}\AgdaSpace{}%
\AgdaSymbol{(}\AgdaDatatype{S}\AgdaSpace{}%
\AgdaOperator{\AgdaPostulate{↓₀}}\AgdaSymbol{)}\AgdaSpace{}%
\AgdaBound{X}\AgdaSpace{}%
\AgdaBound{i}\<%
\\
\>[2]\AgdaFunction{S-ε⁻¹}\AgdaSpace{}%
\agdamark{nil}{\AgdaDottedPattern{\AgdaSymbol{.}}\AgdaDottedPattern{\AgdaInductiveConstructor{[]}}}%
\>[25]\agdamark{leaf}{\AgdaInductiveConstructor{leaf}}%
\>[59]\AgdaSymbol{=}\AgdaSpace{}%
\AgdaInductiveConstructor{leaf}\<%
\\
\>[2]\AgdaFunction{S-ε⁻¹}\AgdaSpace{}%
\agdamark{cons}{\AgdaDottedPattern{\AgdaSymbol{.(}}\AgdaDottedPattern{\AgdaBound{p}}\AgdaSpace{}%
\AgdaDottedPattern{\AgdaOperator{\AgdaInductiveConstructor{∷}}}\AgdaSpace{}%
\AgdaDottedPattern{\AgdaBound{iₗ}}\AgdaSpace{}%
\AgdaDottedPattern{\AgdaOperator{\AgdaFunction{++}}}\AgdaSpace{}%
\AgdaDottedPattern{\AgdaBound{iᵣ}}\AgdaDottedPattern{\AgdaSymbol{)}}}%
\>[25]\agdamark{pivot}{\AgdaSymbol{((}\AgdaBound{u}\AgdaSpace{}%
\AgdaOperator{\AgdaInductiveConstructor{\textasciicircum{}}}\AgdaSpace{}%
\AgdaBound{iₗ}\AgdaSpace{}%
\AgdaOperator{\AgdaInductiveConstructor{┌}}\AgdaSpace{}%
\AgdaBound{p}\AgdaSpace{}%
\AgdaOperator{\AgdaInductiveConstructor{┐}}\AgdaSpace{}%
\AgdaBound{v}\AgdaSpace{}%
\AgdaOperator{\AgdaInductiveConstructor{\textasciicircum{}}}%
\AgdaBound{iᵣ}\AgdaSymbol{)}\AgdaSpace{}%
\AgdaBound{pf₁}\AgdaSpace{}%
\AgdaBound{pf₂}\AgdaSymbol{)}}%
\>[59]\AgdaSymbol{=}\<%
\\
\>[2][@{}l@{\AgdaIndent{0}}]%
\>[5]\AgdaSymbol{((}\agdamark{helpl}{\AgdaFunction{iₗ<i}}\AgdaSpace{}%
\AgdaOperator{\AgdaInductiveConstructor{,}}\AgdaSpace{}%
\AgdaBound{u}\AgdaSymbol{)}\AgdaSpace{}%
\AgdaOperator{\AgdaInductiveConstructor{┌}}\AgdaSpace{}%
\AgdaBound{p}\AgdaSpace{}%
\AgdaOperator{\AgdaInductiveConstructor{┐}}\AgdaSpace{}%
\AgdaSymbol{(}\agdamark{helpr}{\AgdaFunction{iᵣ<i}}\AgdaSpace{}%
\AgdaOperator{\AgdaInductiveConstructor{,}}\AgdaSpace{}%
\AgdaBound{v}\AgdaSymbol{))}\AgdaSpace{}%
\AgdaBound{pf₁}\AgdaSpace{}%
\AgdaBound{pf₂}\AgdaSpace{}%
\AgdaKeyword{where}%
\>[1248I]\agdamark{defhelpl}{\AgdaFunction{iₗ<i}\AgdaSpace{}%
\AgdaSymbol{:}\AgdaSpace{}%
\AgdaSymbol{(}\AgdaBound{iₗ}\AgdaSpace{}%
\AgdaOperator{\AgdaFunction{<♯}}\AgdaSpace{}%
\AgdaBound{p}\AgdaSpace{}%
\AgdaOperator{\AgdaInductiveConstructor{∷}}\AgdaSpace{}%
\AgdaBound{iₗ}\AgdaSpace{}%
\AgdaOperator{\AgdaFunction{++}}\AgdaSpace{}%
\AgdaBound{iᵣ}\AgdaSymbol{)}}\AgdaSpace{}%
\AgdaSymbol{;}\AgdaSpace{}%
\agdamark{defhelpr}{\AgdaFunction{iᵣ<i}\AgdaSpace{}%
\AgdaSymbol{:}\AgdaSpace{}%
\AgdaSymbol{(}\AgdaBound{iᵣ}\AgdaSpace{}%
\AgdaOperator{\AgdaFunction{<♯}}\AgdaSpace{}%
\AgdaBound{p}\AgdaSpace{}%
\AgdaOperator{\AgdaInductiveConstructor{∷}}\AgdaSpace{}%
\AgdaBound{iₗ}\AgdaSpace{}%
\AgdaOperator{\AgdaFunction{++}}\AgdaSpace{}%
\AgdaBound{iᵣ}\AgdaSymbol{)}}\<%
\end{code}%
\begin{tikzpicture}[remember picture,overlay]
  \coordinate (desc box height) at ([yshift=8mm]leaf.north);
  \node[description node,anchor=base] (desc pattern) at ([xshift=23.0mm]leaf.east |- desc box height) {(1) Pattern match};
  \draw[description arrow] (desc pattern.south) |- (leaf.east);
  \draw[description arrow] (desc pattern.south) -- (pivot.north -| desc pattern.south);
  \node[description node,anchor=base] (inversion) at ([xshift=10.0mm]nil.east |- desc box height) {(2) Inversion};
  \draw[description arrow] (inversion.south) |- (nil.east);
  \draw[description arrow] (inversion.south) -- (cons.north -| inversion.south);
  \node[description node,anchor=base] (helperdefs) at (defhelpr |- desc box height)
      {\begin{minipage}[t]{.25\textwidth}
        (3) Subproofs that bags $i_l$, $i_r$
            are smaller (only their type
            is shown, the definition is hidden)
        \end{minipage}%
      };
  \coordinate (branching point) at (cons.east -| helperdefs.south);
  \draw[description arrow] (helperdefs.south) -- (branching point) -| (defhelpl.north);
  \draw[description arrow] (helperdefs.south) -- (branching point) -- (defhelpr.north);
\end{tikzpicture}%
%
The definition of \af{S-ε⁻¹} follows the following scheme:
(1) Pattern match on the value of type \af{S}~\(X\ i\);
(2) by \emph{inversion}~\cite{dybjerInductiveFamilies1994}, this will refine the original index (seen here as \emph{dot patterns} \cite{dotpatterns});
(3) prove that the indices in the functorial positions are smaller than the original, now refined, outer index. The map \af{S-ε⁻¹} embeds \af{S} into \af{S↓} and thus witnesses that the type constructor \af{S} is well-founded.

\section{Ranked Coalgebras and Termination Proofs}\label{sec:ranked-coalgebras}
We have seen how recursive coalgebras in categories $\C^I$ of indexed families give rise to intrinsic total correctness of recursive algorithms. When the interest is only in proving \emph{termination}, it is usually sufficient and conceptually easier to work with coalgebras in the base category $\C$. For that purpose, we introduce next the notion of \emph{ranked coalgebra}. Informally, a ranked coalgebra is a coalgebra in $\C$ that has the proof of its well-foundedness (i.e.\ termination) baked into its definition, by associating to every state a \emph{rank} from the set $I$ and requiring that transitions strictly decrease the rank w.r.t.\ the relation $<$ on $I$. Since $<$ is well-founded, the rank cannot decrease indefinitely, so the coalgebra is well-founded. Ranked coalgebras thus form a coalgebraic abstraction of the familiar technique for termination proofs of programs based on ranking functions~\cite{cookPodelskiRybalchenko}. On a technical level, we exploit that since $\C$ has coproducts (\autoref{indexCat}), families in $\C^I$ can be internalized in~$\C$. 

\begin{definition}[Ranked Coalgebra]
Let $F\colon \C\to \C$ be an endofunctor. A \emph{ranked family} for $F$ is given by a family $(C_i)_{i\in I}$ of objects of $\C$ and a family of morphisms $(c_i\colon C_i\to F(\coprod_{j<i} C_j))_{i\in I}$. Every ranked family induces an $F$-coalgebra $(C,c)$ with carrier $C=\coprod_{i\in I} C_i$ and structure
\begin{equation}\label{eq:induced-coalgebra} c\colon \coprod_{i\in I} C_i \xto{\coprod_{i\in I} c_i} \coprod_{i\in I}F(\coprod_{j<i} C_j) \xto{[F\inj_{<i}]_{i\in I}} F(\coprod_{i\in I} C_i), \end{equation}
where $\inj_{<i}=[\inj_j]_{j<i}$.
An $F$-coalgebra is \emph{ranked} if it is induced by some ranked family. 
\end{definition}

Our main result on ranked coalgebras is that they are \emph{intrinsically} recursive and well-founded:
\begin{theorem}[Intrinsic Recursivity / Well-Foundedness]\label{thm:well-founded-recursive-coprod} Let $F\colon \C\to \C$ be a functor.
\begin{enumerate}
\item\label{thm:well-founded-recursive-coprod:1} Every ranked $F$-coalgebra is recursive.
\item\label{thm:well-founded-recursive-coprod:2} Every ranked $F$-coalgebra is well-founded.
\end{enumerate}
\end{theorem}
 In essence, this an analogue of \autoref{allCausalRecursive} where the indexing occurs at the level of the carrier of the coalgebra rather than at the level of the underlying category. Below we sketch the proof of part \ref{thm:well-founded-recursive-coprod:1}. The idea is to reduce this statement to its indexed counterpart given by \autoref{allCausalRecursive}\ref{allCausalRecursive:1} and to apply the \emph{generalized powerset construction}~\cite{generalizeddeterminization} for coalgebras.

\begin{proof}[Proof Sketch for \ref{thm:well-founded-recursive-coprod:1}]
 Let $(c_i\colon C_i\to F(\coprod_{j<i} C_j))_{i\in I}$ be a ranked family. To prove that its induced coalgebra \eqref{eq:induced-coalgebra} is recursive, we consider the following functor $G$ on $\C^I$ defined by composition, where $\coprod$ is the coproduct functor and $\Delta$ is the diagonal functor given by $(\Delta X)_i=X$:
  \[
    G\equiv\big(
    \begin{tikzcd}
      \C^I
      \arrow{r}{\coprod}
      & \C
      \arrow{r}{F}
      & \C
      \arrow{r}{\Delta}
      & \C^I
    \end{tikzcd}
    \big),
    \qquad
    (GX)_i = (\Delta F \coprod X)_i
    = F\coprod X = F \coprod_{j\in I} X_j.
  \]
  The well-founded coreflection of $G$ is given by
  \[
    G_\downarrow\colon \C^I\to\C^I,
    \qquad
    (G_\downarrow X)_i = (G T_{<i} X)_i
    = F \coprod_{j\in I}
    (T_{<i}X)_j 
    = F \coprod_{j < i} X_j,
  \]
  since $(T_{<i}X)_j=X_j$ for $j<i$ and $(T_{<i}X)_j=0$ otherwise. 
  Thus the given family $(c_i)$ is a $G_\downarrow$-coalgebra
  \[
    c\colon C\to G_\downarrow C
    \qquad\text{ in }\C^I,
  \]
 and by \autoref{allCausalRecursive}\ref{allCausalRecursive:1} this coalgebra is recursive. Composition with the component $\epsilon_{G,C}$ of the natural transformation $\epsilon_{G}\colon G_\downarrow  \to G$ yields a $\Delta F\coprod$-coalgebra
  \[
    \epsilon_{G,C}\cdot c\colon
    \begin{tikzcd}
      C
      \arrow{r}{c}
      & G_\downarrow C
      \arrow{r}{\epsilon_{G,C}}
      & G C
      = \Delta F\coprod C,
    \end{tikzcd}
  \]
  which is again recursive~\cite[Prop.~1]{eppendahlFixedPointObjects2000}. Since the coproduct functor $\coprod$ is left adjoint to the diagonal functor $\Delta$, we get by adjoint transposition the $F$-coalgebra
  \[
    \begin{tikzcd}[column sep=40]
      \coprod C
      \arrow{r}{(\epsilon_{G,C}\cdot c)^\sharp}
      & F \coprod C.
    \end{tikzcd}
  \]
  This is precisely the coalgebra \eqref{eq:induced-coalgebra} induced by the ranked family $(c_i)$. To see that it is recursive, we note that 
  the last step is an instance of the \emph{generalized powerset construction}~\cite{generalizeddeterminization}. One can show that this construction always preserves recursivity; see \autoref{sec:rec-preservation} for more details. 
\end{proof}

\begin{proofappendix}{thm:well-founded-recursive-coprod}
\begin{enumerate}
\item Let $(c_i\colon C_i\to F(\coprod_{j<i} C_j))_{i\in I}$ be a ranked family. To prove that its induced coalgebra \eqref{eq:induced-coalgebra} is recursive, we consider the following functor $G$ on $\C^I$ defined by composition, where $\coprod$ is the coproduct functor and $\Delta$ is the diagonal given by $(\Delta X)_i=X$:
  \[
    G\equiv\big(
    \begin{tikzcd}
      \C^I
      \arrow{r}{\coprod}
      & \C
      \arrow{r}{F}
      & \C
      \arrow{r}{\Delta}
      & \C^I
    \end{tikzcd}
    \big),
    \qquad
    (GX)_i = (\Delta F \coprod X)_i
    = F\coprod X = F \coprod_{j\in I} X_j.
  \]
  Its well-founded coreflection is given by
  \[
    G_\downarrow\colon \C^I\to\C^I,
    \qquad
    (G_\downarrow X)_i = (G T_{<i} X)_i
    = F \coprod_{j\in I}
    (T_{<i}X)_j 
    = F \coprod_{j < i} X_j,
  \]
  since $(T_{<i}X)_j=X_j$ for $j<i$ and $(T_{<i}X)_j=0$ otherwise. 
  Thus the given family $(c_i)$ is a morphism
  \[
    c\colon C\to G_\downarrow C
    \qquad\text{ in }\C^I,
  \]
  and thus a recursive coalgebra (\autoref{allCausalRecursive}). By \autoref{recursiveNat}, its composition with the component $\epsilon_{G,C}$ of the natural transformation $\epsilon_{G}\colon G_\downarrow  \to G$ yields a recursive $\Delta F\coprod$-coalgebra
  \[
    \epsilon_{G,C}\cdot c\colon
    \begin{tikzcd}
      C
      \arrow{r}{c}
      & G_\downarrow C
      \arrow{r}{\epsilon_{G,C}}
      & G C
      = \Delta F\coprod C.
    \end{tikzcd}
  \]
  By \autoref{determinizationRecursive} applied to the adjunction $\coprod\dashv
  \Delta$ (where $\D:=\C^I$), the coalgebra
  \[
    \begin{tikzcd}[column sep=40]
      \coprod C
      \arrow{r}{(\epsilon_{G,C}\cdot c)^\sharp}
      & F \coprod C,
    \end{tikzcd}
  \]
is also recursive. This is precisely the coalgebra \eqref{eq:induced-coalgebra} induced by the ranked family $(c_i)$.
\item Let $(c_i\colon C_i\to F(\coprod_{j<i} C_j))_{i\in I}$ be a ranked family and let $(C,c)$ denote the induced coalgebra given by \eqref{eq:induced-coalgebra}. To prove that $(C,c)$ is well-founded, suppose that $m\colon (S,s)\monoto (C,c)$ is a cartesian subcoalgebra. For each $i\in I$ we prove that 
\begin{equation}\label{eq:pf-goal} \exists n_i\colon C_i\to S.\, \inj_i = m\cdot n_i. \end{equation}
This implies that the monomorphism $m$ is a split epimorphism:
\[ m\cdot [n_i]_{i\in I} = [\inj_i]_{i\in I}= \id_C, \]
whence an isomorphism as required. 

We prove \eqref{eq:pf-goal} by well-founded induction. Let $i\in I$ and suppose that \begin{equation} \exists n_{j}\colon C_{j}\to S.\, \inj_{j} = m\cdot n_{j} \qquad\text{for all $j<i$.} \end{equation}
Then the outside and the right-hand part of the diagram below commute, and the universal property of the pullback yields a unique morphism $n_i$ making the remaining parts commute, proving \eqref{eq:pf-goal}.
\[ 
\begin{tikzcd}
C_i \ar{rrr}{c_i} \ar[dashed]{dr}{n_i} \ar[bend right=20]{ddr}[swap]{\inj_i} & & & F(\coprod_{j<i} C_j) \ar{dl}[swap]{F[n_{j}]_{j< i}} \ar[bend left=20]{ddl}{F\inj_{< i}} \\
& S \pullbackangle{-45} \ar{r}{s} \ar[tail]{d}[swap]{m}  & FS \ar{d}{Fm} & \\
& C \ar{r}{c} & FC & 
\end{tikzcd} \qedhere
\]
\end{enumerate}
\end{proofappendix}

For $\C=\Set$, ranked coalgebras admit a simple characterization by \emph{ranking functions}, which uses the equivalence  $\Set^I\simeq \Set/I$ \eqref{eq:familySlice} and the relation between coalgebras and canonical graphs.

\begin{definition}[Ranking Function]
Let $F \colon \Set\to \Set$ be a functor.
A \emph{ranking function} for an $F$-coalgebra $(C,c)$ is a function $r\colon C\to I$ satisfying the condition
\[ \forall x\in C.\, c(x)\in F(\{ y\in C\mid r(y)<r(x) \}). \]   
\end{definition}
This captures on an abstract level the requirement that transitions strictly decrease the rank.

\begin{theorem}[Characterization of Ranked Coalgebras in $\Set$]\label{thm:ranked-coalg-set}
Suppose that $F\colon \Set\to\Set$ preserves wide intersections. Then for every $F$-coalgebra $(C,c)$, we have that:
\[ \text{$(C,c)$ is ranked} \iff \text{$(C,c)$ has a ranking function} \iff \text{$(C,c)$ is well-founded.} \]
\end{theorem}

\begin{proof}
 The first statement implies the third by \autoref{thm:well-founded-recursive-coprod}. To prove that the third statement implies the second, suppose that $(C,c)$ is a well-founded coalgebra. By \autoref{prop:wf-set}, this means precisely that the relation 
\[ \mathord{<}\,\seq\, C\times C \qquad\text{given by}\qquad y < x \iff y\in \tau_C(c(x))   \]
is well-founded. Moreover, the identity map
\[ r=\id\colon C\to C \]     
is a ranking function for $(C,c)$ since, for every $x\in C$,
\[ c(x)\in  F(\tau_C(c(x))) = F(\{ y \in C \mid y<x \}) = F(\{ y\in C \mid r(y)<r(x) \}). \]
Finally, the second statement implies the first: if $r\colon C\to I$ is a ranking function for $(C,c)$, then for each $i\in I$ the map $c\colon C\to FC$ restricts to a map \[c_i\colon r^{-1}(i)\to F(r^{-1}({<}i)) = F(\coprod_{j<i} r^{-1}(j)).\] Then $(c_i)_{i\in I}$ is a ranked family whose induced coalgebra is $(C,c)$.
\end{proof}

\begin{example}[Quicksort]
To prove termination of Quicksort, we consider the coalgebra $c\colon Z^*\to 1+Z^*\times Z\times Z^*$ of \eqref{eq:sort-coalgebra} and the function $r\colon Z^*\to \Nat$ given by $r(w)=|w|$. Equipping $\Nat$ with the usual order $<$, this is a ranking function for $(C,c)$: since $|\epsilon|=0$ and $|w_{\sle p}|, |w_{\sgt p}|<|pw|$, we have
\begin{align*}c(\epsilon)=\inl\in 1=F(\emptyset)=F(\{y\in Z^* \mid r(y)<r(\epsilon)),\\ c(pw)=w_{\sle p} \, p \, w_{\sgt p}\in F(\{y\in Z^* \mid r(y)<r(pw) \}).
\end{align*}
We conclude from \autoref{thm:ranked-coalg-set} that $(C,c)$ is well-founded, hence the Quicksort recursion terminates. This reasoning is essentially a more principled account of the argument in the proof of \autoref{prop:sort-recursive}, where we have already used the above ranking function implicitly.
\end{example}

We have thus demonstrated that the use of ranking functions for proving program termination, standardly applied to transition graphs of programs, extends smoothly from graphs to general coalgebras. The same holds for more advanced ranking techniques. Let us illustrate this for the powerful method of \emph{disjunctive well-foundedness}~\cite{cookPodelskiRybalchenko,pr04}, which brings additional flexibility to termination arguments by working with several ranking functions at the same time. Given a graph $c\colon C\to \Pow C$ and a finite family of functions $r_k\colon C\to I_k$ ($k\in K$), where each index set $I_k$ is equipped with a well-founded relation $<_k$, the following implication holds~\cite[Thm.~1]{pr04}:
\begin{equation}\label{eq:disjunctive-wf} \big[\forall y\in c^{+}(x).\,\exists k\in K.\, r_k(y)<_kr_k(x)\big] \qquad\implies\qquad \text{$(C,c)$ is well-founded}.\end{equation}
Here $c^{+}$ denotes the transitive closure of $c$ (i.e.\ $y\in c^{+}(x)$ iff there exists a non-empty path from $x$ to~$y$). The proof of \eqref{eq:disjunctive-wf} relies on Ramsey's theorem. We obtain the following coalgebraic generalization:

\begin{theorem}[Coalgebraic Disjunctive Well-Foundedness]
Suppose that $F\colon \Set\to\Set$ preserves wide intersections. For every coalgebra $(C,c)$ and every finite family of functions $r_k\colon C\to I_k$ ($k\in K$), where $I_k$ is equipped with a well-founded relation $<_k$, the following implication holds:
\[ \big[\forall x\in C.\, \forall y\in (\tau_C\cdot c)^{+}(x).\,\exists k\in K.\, r_k(y)<_kr_k(x)\big] \qquad\implies\qquad \text{$(C,c)$ is well-founded}.\]   
\end{theorem}

\begin{proof}
Since $(C,c)$ is well-founded iff its canonical graph $(C,\tau_C\cdot c)$ is well-founded (\autoref{prop:wf-set}), this statement is immediate from \eqref{eq:disjunctive-wf}.
\end{proof}

\section{Case Studies}\label{sec:case-studies}
We have illustrated how the framework of well-founded functors enables an intrinsically correct \& terminating definition of Quicksort. In the following we showcase the scope of our coalgebraic toolkit by presenting several additional applications, going beyond sorting algorithms.

\subsection{Euclidian Algorithm}\label{sec:euclidian}
The well-known Euclidian algorithm~\cite{knuth97} computes the greatest common divisor of a given pair of
natural numbers $m$ and $n$, using the recursive formula
\[
  \gcd\colon \N\times \N\to \N,
  \qquad
  \gcd(m,n)
  = \begin{cases}
    m&\text{if }n=0,\\
    \gcd(n,m\modop n)&\text{if }n\neq 0.
  \end{cases}
\]
The algorithm itself has been modelled coalgebraically by \citet{taylor99} and in its extended form by \citet{jeanninWellfoundedCoalgebrasRevisited2017}. We give a refined treatment that entails its intrinsic correctness. The overall modus operandi for the coalgebraic correctness proof is very similar to the case of Quicksort laid out in \autoref{sec:prelim} and \ref{sec:well-founded-functors}. First, we identify a suitable functor that models the recursive branching. The two cases in the recursive computation of $\gcd$ are captured by the functor
\[
  F\colon \Set\to\Set,
  \qquad
  FX = \N + \N \times X.
\]
The left coproduct component $\N$ describes the base case of $\gcd$ and the second coproduct component
$\N\times X$ consists of the arguments $X$ to the recursive call and an additional
natural number keeping track of the quotient \(\floor{m/n}\), which is used for correctness.
With this functor, the recursive computation of $\gcd$ can be translated into the following
$F$-coalgebra on $\NgN$:
\begin{equation}\label{eq:coalg-euclid}
  c\colon \NgN \to \underbrace{F (\NgN)}_{\mathclap{\N + \N\times (\NgN)}},
  \qquad
  c(m,n)
  = \begin{cases}
    \inl (m)&\text{if }n = 0,\\
    \inr (\floor{m/n}, (n,m\modop n))&\text{if }n \neq 0.\\
  \end{cases}
\end{equation}
We turn the carrier of the above coalgebra
into a family by considering
the identity $\id\colon \NgN\to \NgN$ as the ranking function, yielding an
object $\Set/\NgN$.
Under the
equivalence $\Set^{\NgN} \cong \Set/\NgN$ \eqref{eq:familySlice} to ${\NgN}$-indexed
families, the carrier $C\in \Set^{\NgN}$ is the $\NgN$-indexed family 
\[
C\in \Set^{\NgN},
\qquad
C = (1)_{(m,n)\in {\NgN}}
\tag{\agdarefcustom{Coalgebra carrier}{GCD}{C}{}},
\]
given by singleton sets at each index. In general, the coalgebra carrier in $\Set^I$ will be the family of singletons
whenever the index set coincides with the input of the desired algorithm
(another such algorithm is \CYK in \autoref{sec:cyk-parsing}). For the
algebra carrier, we use the following family of sets:
\[
A\in \Set^{\NgN},
\qquad
A(m,n)
= \set{g \in \N \;:\; g\mathrel{\mid} m\text{ and } g\mathrel{\mid} n\text{ and } \forall d\colon d\mathrel{\mid} m,~d\mathrel{\mid} n\to d\mathrel{\mid} g}
\tag{\agdarefcustom{Algebra carrier}{GCD}{A}{}}.
\]
At index $(m,n)\in \NgN$, the set $A(m,n)$ describes precisely what we want to construct: an arbitrary
number $g\in \N$ that divides both $m$ and $n$, and moreover, is divided by any
other common divisor of $m$ and $n$. Here, we do not need the fact the greatest
common divisor is unique.
\begin{observation}[Intrinsic Correctness, \agdaref{GCD}{observation}{}]\label{obsGcdCorret}
  Every morphism \[h\colon C \to A \qquad\text{in }\Set^{\NgN}\] returns
  the greatest common divisor.
  For $m,n\in \N$, the function $h_{m,n}\colon C(m,n)\to A(m,n)$ can be invoked
  with $*\in C(m,n)=1$ and returns the desired number $h_{m,n}(*) \in A(m,n)$.
\end{observation}
As in the case for Quicksort, it now suffices to define any
$\NgN$-indexed map $h\colon C\to A$. To this end, we turn the above $F$ into
a functor $G\colon \Set^{\NgN} \to \Set^{\NgN}$. Given $X\in \Set^{\NgN}$,
we define $GX\colon \N\times\N\to \Set$ by the following constructors (instructively called $\inl$ and $\inr$):
\[
  \begin{array}{llll}
    \text{for all } m\in \N
    &\text{ we put }
    &\inl(m) &\in G(X)(m,0)
    \\
    \text{for all } q,m,n\in \N
    , m > n,
    x \in X(m,n)
    &\text{ we put }
    &\inr(q,x) &\in G(X)(q\cdot m + n,m)
  \end{array}
  \tag{\agdarefcustom{$G\colon \Set^{\NgN}\to\Set^{\NgN}$}{GCD}{G}{In the formalization, 
  the proposition $a>b$ is an explicit parameter of $\inr$}}
\]
Via the equivalence $\Set/\NgN\cong \Set^\NgN$ \eqref{eq:familySlice}, the
functor $G$ corresponds to the following functor on the slice category (only listed here for completeness):
\begin{align*}
  &\bar F\colon \Set/\NgN\to \Set/\NgN
  &
  \bar F(X,\fpair{r_1,r_2}) = 
  \N + \N \times \set{x\in X\mid r_1(x) > r_2(x)}
  \subseteq FX
  \\
  &\bar r\colon \bar F(X,\fpair{r_1,r_2})\to \NgN
  &
  \bar r(\inl(a)) = (a, 0)
  \qquad
  \bar r (\inr (q,x)) = (q\cdot r_1(x) + r_2(x), r_1(x))
\end{align*}
Thus far, we have used the index $I:= \NgN $ to intrinsically express the functional correctness of the GCD algorithm in \(\Set^I\).
The functor \(G\) we have defined is also well-founded, given the well-founded relation
\[
  (m_1,n_1)
  < 
  (m_2,n_2)
  \quad
  \Longleftrightarrow
  \quad
  n_1< n_2.
  \tag{\agdarefcustom{Relation for GCD}{GCD}{<₂-wellFounded}{}}
\]
\begin{lemma}[{\agdaref{GCD}{G-ε⁻¹}{witnesses the well foundedness}}]
  \label{lem:barFHwellFounded}
  The functor $G\colon \Set^{\NgN}\to \Set^{\NgN}$ is well-founded. 
\end{lemma}
\begin{proofappendix}{lem:barFHwellFounded}
For the textbook proof, we work with the corresponding functor $\bar F$ on the slice category.
We use the criterion of \autoref{prop:lifted-functor-well-founded}. Fix $(m,n)\in \N\times \N$ and $(X,r)\in \Set/{\NgN}$. We need to verify that:
\[\text{Every element of  $\bar r^{-1}(m,n)$ lies in $\N+\N\times r^{-1}({<}(m,n))$}.\]
If $n=0$, then the only element sent by $\bar r$ to $(m,0)$ is $\inl(m)$, which lies in $\N+\N\times r^{-1}({<}(m,0))=\N$. If $n>0$, then $\inr(q,x)$ is sent to $(m,n)$ iff $m=r_1(x)\cdot q+r_2(x)$ and $n=r_1(x)$. Since $r_2(x)<r_1(x)$ by definition of $\bar F_1$, we have $r_2(x)<n$ and therefore $(r_1(x),r_2(x))<(m,n)$ in the well-founded order in $\N\times \N$. This shows that $x$ lies in $r^{-1}({<}(m,n))$, so $\inr(q,x)$ lies in $\N\times r^{-1}({<}(m,n)) \seq \N + \N\times r^{-1}({<}(m,n))$.
\end{proofappendix}

The coalgebra $c\colon \NgN\to F(\NgN)$ \eqref{eq:coalg-euclid} lifts to a
$G$-coalgebra $C\to GC$ in $\Set^{\NgN}$ (\agdarefcustom{Coalgebra for GCD}{GCD}{coalg}{}), which is
equivalent to a $\bar F$-coalgebra in $\Set/{\NgN}$.
The $G$-algebra structure on $A$ is defined by:
\begin{equation}
  \alg\colon GA\to A
  \qquad
  \begin{array}{l@{}l@{\,}l}
    \alg_{(a,0)} & (\inl(a)) &= a \\
    \alg_{(q\cdot a + b,a)}&(\inr(q,g)) &= g
    \quad\text{where }g\in A(a,b)
  \end{array}
  \label{eq:alg-euclid}
\end{equation}
\begin{lemma}[{\agdaref{GCD}{alg}{}}] \label{lem:alg-well-typed}
  The $G$-algebra in \eqref{eq:alg-euclid} is well-typed.
\end{lemma}
\begin{proof}
The algebra in $\Set^{\NgN}$ is a family of maps $\alg_{m,n}\colon G(A)(m,n)\to A(m,n)$ for $(m,n)\in \NgN$.
We show that the pattern match is well-defined and that the image indeed is in $A(m,n)$.
By
the definition of $G$, we know that $(m,n)$ is of the form $(a,0)$ or $(q\cdot
a+b, a)$ (with $a > b \ge 0$) and each of these cases admits only one case for the
argument of type $G(A)(m,n)$:
\begin{enumerate}
\item In the case $\inl(a)$, the number $a$ is the greatest common divisor of 
$m=a$ and $n=0$.
\item In the case $\inl(q,g)$ we have $g\in A(a,b)$ and
\[
  m = q\cdot a + b
  \qquad
  n = a.
\]
The parameter provides that $g$ is a greatest common divisor of $a$ and $b$ and we now need to show
that it is also the greatest common divisor of $m$ and $n$.
\begin{itemize}
\item Since $g$ divides $a$ and $b$, it also divides $m=q\cdot a + b$.
\item Any other common divisor $d$ of $m$ and $n=a$, also divides $q\cdot a$ and so also $b = m - q\cdot a$.
Hence, $d$ divides $g$ by $g\in A(a,b)$.
\qedhere
\end{itemize}
\end{enumerate}
\end{proof}

We are ready to prove the intrinsic correctness of the Euclidian algorithm:

\begin{theorem}[Correctness of the Euclidian Algorithm, \agdaref{GCD}{gcd}{}]\label{thm:euclid-correctness}
The Euclidian algorithm computes for any two input numbers $m$ and $n$ the greatest common divisor of $m$ and $n$.
\end{theorem}

\begin{proof}Since the functor $G$ is well-founded, the $G$-coalgebra $\coalg\colon C\to GC$ is recursive (\autoref{allCausalRecursive}). Hence it induces a unique coalgebra-to-algebra morphism $h \colon C \to A$ in \(\Set^{\NgN}\).
By \autoref{obsGcdCorret} we conclude that this map yields the greatest common divisor.
  \end{proof}

\takeout{
The above reasoning carries over to the \emph{extended} Euclidian algorithm, which in addition to the greatest common divisor $g$ of $m,n\in\N$ also computes integers $s,t\in \Z$ such that
\[
  s\cdot m + t\cdot n = g.
\]
These are precisely the witnesses $s,t$ satisfying $s\cdot k + t\cdot \ell = 1$
from the definition of $\N\coprime \N$ in the above algebra $(A,a)$:
\[
  s\cdot k + t\cdot \ell = 1
  \quad\Longrightarrow\quad
  s\cdot \underbrace{g\cdot k}_{m} ~+~ t\cdot \underbrace{g\cdot \ell}_{n} = g
\]
  One obtains the extended Euclidian algorithm directly by internalizing these
  witnesses $s,t\in \Z$ in the carrier of the
  $F$-algebra: we extend the above algebra $A$ to
  \[
    B := \set{(g,k,\ell,s,t)\in \N\times\N\times\N\times \Z\times \Z \mid s\cdot k + t\cdot \ell = 1}
  \]
  \text{with rank}
  \[
    r_B\colon B\to \N\times \N,
    \qquad
    r_B(g,k,\ell,s,t) = r_A(g,k,\ell).
  \]
  So in particular, $(k,\ell) \in \N\coprime\N$. The $F$-algebra structure on $B$ is:
  \[
    b\colon \N+\N\times B \to B,
    \qquad
    \begin{array}{rl}
    b(\inl(m))
    &= (a(\inl(m)), (1,0)),
    \\
    b(\inr(q,k,\ell,s,t))
    &= (a(\inr(g,k,\ell)), t, s-t\cdot q).
    \end{array}
  \]
  On the first three components of the algebra, the computation by $b$ is identical to $a$ above.
  The computation of the witnesses $s,t$ reflects the witnesses from the
  well-typedness of $a$ above~\eqref{aIsWellTyped}.
  Likewise, $b$ inherits the property of preserving the rank, so we obtain a coalgebra-to-algebra morphism
  \[
    \eea\colon (\N\times\N,\id)
    \longrightarrow (B,r_B)
    \qquad\text{in }\Set/I
  \]
  that necessarily computes the extended Euclidian algorithm.
\begin{corollary}
  The unique coalgebra-to-algebra morphism $(\N\times \N,c)\to (B,b)$
  computes the extended Euclidian algorithm.
\end{corollary}
}

\subsection{CYK Parsing}
\label{sec:cyk-parsing}
The \CYK algorithm~\cite{hopcroft1979} determines whether a given input string
can be parsed by a context-free grammar in Chomsky normal form (CNF). In the following, fix a context-free grammar with a set~$V$ of non-terminals and a set $\Sigma$ of terminals. The grammar is in CNF if all its rules are of the form
\[ P\to QT~(P,Q,T\in V) \qquad\text{or}\qquad P\to \sigma~(\sigma\in \Sigma).  \]  
We write $P\derive w$ if the word $w$ over $\Sigma$ is derivable from the non-terminal $P$ (\agdarefcustom{CNF Semantics}{Context-Free-CNF-Grammar}{Semantics}{}). Note that only non-empty words are derivable, that is, $w\in \Sigma^+$. The \CYK algorithm computes the map
\[ \CYK \colon \Sigma^+\to \Powf(V), \qquad w\mapsto \{ P\in V\mid P\derive w \}, \]
where $\Powf$ is the finite powerset functor.
With this map, parsability is easy to decide: if the grammar has a starting symbol $S\in V$, a word $w$ is parsable iff $S\in \CYK(w)$.  
The map \CYK itself is computed via the following recursive procedure:
\begin{enumerate}
\item $\CYK(\sigma)$, for $\sigma\in \Sigma$, is the set of those non-terminals $P$ with a
  rule $P\to \sigma$.
\item $\CYK(w)$, for $w\in \Sigma^+$ of length at least 2, is the set of all non-terminals $P$ for which there exists a rule $P\to QT$ and a decomposition $w=uv$ ($u,v\in \Sigma^+$) with $Q\in \CYK(u)$ and $T\in \CYK(v)$.
\end{enumerate}
When considering the call graph of this recursion, e.g.\ for $abcd \in
\Sigma^4$, then $\CYK(bc)$ is recursively called multiple times.
Thus, the usual presentation of the \CYK~algorithm uses
dynamic programming and computes the recursion bottom up, by
iterating over all subwords $v$ of $w$ and computing $\CYK(v)$. The result is
then saved in an array so that the result can be used directly when computing
$\CYK(v')$ for longer subwords $v'$. This algorithm is cubic in
the length of $w$. One can achieve the same run time in the recursive version by memoisation of intermediate results (in an array of the same type as in the iterative version) and thus short-circuiting calls for the same subword. 

In order to capture $\CYK$ and its intrinsic correctness in our coalgebraic framework, we stick with the recursive formulation and index over the set of non-empty words:
\[
  I := \Sigma^+ \qquad \text{with}\qquad v<w \iff |v|<|w|.
\]
We model \CYK within the category $\Set^{\Sigma^+}$ of $\Sigma^+$-indexed families. In the present setting, unlike for Quicksort (\autoref{sec:quicksort-correctness}) and the Euclidian algorithm (\autoref{sec:euclidian}), this category is a more convenient foundation than the equivalent slice category $\Set/{\Sigma^+}$. We consider the following families:
\[
  C\in \Set^{\Sigma^+},
  \quad
  C_w := 1,
  \qquad
  A\in \Set^{\Sigma^+},
  \quad
  A_w := \set[\big]{\set{P\in V\mid P\derive w}}.
  \tag{\agdarefcustom{C, A}{CYK}{A}{Both the definitions of $C$ and $A$}}
\]
The family $C$ is constantly the singleton set, and the family $A$ also contains singletons only.

\begin{observation}[Intrinsic Correctness]\label{obsCYKCorrect}
Every morphism
\[
  h\colon C\longrightarrow A
  \qquad
  \text{in }
  \Set^{\Sigma^+}
\]
provides for each $w\in \Sigma^+$ a map
$h_w\colon 1\to A_w$, that is, an element of $A_w$. Therefore, for each $w$,
the map $h_w$ necessarily yields 
the set of all non-terminals from which $w$ is derivable.
\end{observation}

The recursive step of the \CYK algorithm, which considers all possible decompositions of the input word, is modelled by the functor 
  \[
    G\colon \Set^{\Sigma^+}\to  \Set^{\Sigma^+}\qquad\text{given by}\qquad (GX)_w = \prod_{\substack{u,v\in\Sigma^{+}\\ uv = w}}
      (X_u\times X_v).
    \tag{\agdarefcustom{Functor\ $G$ for \CYK}{CYK}{F₀}{see also \texttt{F₁} }}
  \]
  In particular, for $\sigma\in \Sigma$, there are no $u,v\in\Sigma^+$ with
  $uv = \sigma$ and so $(GX)_\sigma = 1$. The functor $G$ is well-founded since $(G-)_w$ depends only on input components indexed by words shorter than $w$. 

The constant family $C$ carries a canonical $G$-coalgebra structure (where
$1=\set{*}$):
\[
  c\colon C\to GC,
  \qquad
  c_w\colon 1 
  \rightarrow (GC)_w,
  \qquad
  c_w(*) = ((u,v) \mapsto \underbrace{(*,*)}_{\in C_u\times C_v})
  \in (GC)_w.
  \tag{\agdarefcustom{Coalgebra structure $c$}{CYK}{coalg}{}}
\]
The main work by the $\CYK$ algorithm is performed when combining the results
from the recursion, which is reflected by the $G$-algebra structure
\[
 a\colon GA\to A,\qquad a_w\colon \underbrace{\prod_{\substack{u,v\in\Sigma^{+}\\ uv = w}} A_u\times A_v}_{(GA)_w}\to A_w,
\]
\[
  a_w(p)
  =
  \begin{array}[t]{@{}l@{}}
  \set{P\in V\mid \exists \sigma\in \Sigma: w = \sigma, P\to \sigma  } \\
  \cup \set{P\in V\mid P\to QT, \exists uv=w\colon Q\in \prl(p(u,v)), T\in \prr(p(u,v))}.
  \end{array}
\]
Here, we regard an element $p$ of the product as a dependent function which sends $(u,v)$ with
$uv=w$ to the respective component $A_u\times A_v$, and $\pr_1$ and $\pr_2$ denote the left and right projections of $A_u\times A_v$. The key to the correctness proof for \CYK is the following lemma, which amounts to verifying that~$a$ is a well-typed morphism:
\begin{lemma}[\agdaref{CYK}{alg}{Both the Definition of the algebra and the correctness proof}]
  \label{lemCYKAwelltyped}
  For all $p\in (GA)_w$, the set $a_w(p)$ is equal to the (unique) element of $A_w$, that is,
    \[
   \forall P\in V.\, (P \in a_w(p)
    \Longleftrightarrow
    P\derive w).
 \]
\end{lemma}
\begin{proofappendix}{lemCYKAwelltyped}
We need to show that $a_w(p)$ is the unique element $\set{P\in V\mid P\derive w}$ of $A_w$; that is,
  \[
    P \in a_w(p)
    \quad
    \Longleftrightarrow
    \quad
    P\derive w
    \qquad\text{for all }P\in V.
  \]
To prove ($\Longrightarrow$), let $P\in a_w(p)$. We distinguish two cases, corresponding to the definition of $a_w(p)$:
  \begin{enumerate}
  \item $P\to QT$ and there is $(u,v)$ with $u\,v = w$ and $Q\in \prl(p(u,v))$ and $T\in \prr(p(u,v))$. Note that $\prl(p(u,v))$ and $\prr(p(u,v))$ are the unique elements of the singletons $A_u$ and $A_v$, respectively. By definition of these sets, this means that $Q \derive u$ and $T \derive v$.  Hence
  \[
    P \to QT
    \derive u\,v = w.
  \]
  \item There is $\sigma\in \Sigma $ with $w = \sigma$ and $P\to \sigma $. Then we immediately have $P\derive \sigma = w$.
  \end{enumerate}
To prove ($\Longleftarrow$), suppose that $P\derive w$. We distinguish on the first rule in the derivation.
  \begin{enumerate}
  \item If the derivation $P\derive w$ is of the form $P\to \sigma$ for
  some $\sigma\in \Sigma$, then $\sigma=w$ and $P \in a_w(p)$.
  \item If the derivation $P\derive w$ starts with the rule 
$P\to QT$, then there must be $u,v\in \Sigma^+$ with $uv=w$ and $Q\derive u$ and $T\derive v$.
  Then $p(u,v) \in A_u\times A_v$, and therefore $Q\in \prl(p(u,v))$ and $T\in
  \prr(p(u,v))$ by definition of $A_u$ and $A_v$. This proves $P\in a_w(p)$ by definition of $a_w(p)$.
  \qedhere
  \end{enumerate}
\end{proofappendix}
The proof of this lemma essentially contains the reasoning underlying the usual \enquote{textbook} proof for the correctness of
 \CYK~\cite{hopcroft1979}. The benefit of the present approach is that it isolates this reasoning inside a conceptual statement, in this case the well-typedness of the algebra structure $a$. 
 
The intrinsic correctness of \CYK now easily follows:

\begin{theorem}[Correctness of \CYK, \agdaref{CYK}{CYK}{The morphism arises from the morphism $\epsilon^{-1}$. We do not obtain any property about structure preservation, but do not need that because any morphism $C\to A$ is sufficient by above \autoref{obsCYKCorrect}}]\label{thm:cyk-correct}
  The \CYK algorithm computes for each input word $w\in \Sigma^+$ the set of all non-terminals $P\in
  V$ with $P\derive w$.
\end{theorem}
\begin{proof}
Since $(C,c)$ is recursive, we have the unique coalgebra-to-algebra morphism
  \(h \colon C \to A\) in \(\Set^{\Sigma^+}\).
Note that by definition of the maps $c$ and $a$, the set $\CYK(w)$ of non-terminals recursively computed by the \CYK algorithm is precisely $h_w(*)=(a_w\cdot (Gh)_w \cdot c_w)(*)$. Moreover, by \autoref{obsCYKCorrect}, we know that $h_w(*)=\set{P \in V\mid P\derive w}$. This proves that \CYK is correct. 
\end{proof}

\subsection{Hydra Game}
While the previous case studies have been concerned with intrinsic correctness of algorithms, we next present an example of a non-trivial proof of \emph{termination} that illustrates the technique of coalgebraic ranking functions (\autoref{sec:ranked-coalgebras}). We consider the \emph{Hydra game}~\cite{KirbyParis1982}, a one-player game played on a finite rooted tree. In each round, the player non-deterministically chooses a pair $(l,n)$ of a leaf $l$ and a natural number $n$ and then modifies the tree as follows:

\begin{enumerate}
  \item\label{hydra-1} If $l$ has a grandparent (i.e.\ the parent of $l$ is not the root), add $n$ new leaves to the grandparent.
  \item\label{hydra-2} Delete $l$.
\end{enumerate}

The Hydra game terminates once the player reaches a tree consisting only of the root. Although the tree can grow rapidly during the game (in fact, it can grow faster than any recursive function that is provably total in Peano arithmetic), the game eventually terminates. This is a classic example of a true statement that is unprovable in Peano arithmetic~\cite{KirbyParis1982}. In coalgebraic parlance, the termination result can be phrased as follows. Let $\Trees$ denote the set of finite rooted trees. Form the coalgebra
\begin{equation}\label{eq:coalg-hydra} h\colon \Trees \to \Pow(\Trees) \end{equation}
that sends a tree $T$ to the set of all possible trees emerging from $T$ in one round of the Hydra game, i.e.\ by choosing a pair $(l,n)$ of a leaf of $T$ and a natural number and applying the above modifications \ref{hydra-1} and \ref{hydra-2}. The tree $T$ consisting only of the root has $h(T)=\emptyset$. Termination of the Hydra game then corresponds precisely to the following coalgebraic statement:

\begin{theorem}\label{thm:hydra-well-founded}
The coalgebra $(\Trees,h)$ is well-founded.
\end{theorem}

\begin{proof}
By \autoref{thm:well-founded-recursive-coprod} it is enough to show that $(\Trees,h)$ is a ranked coalgebra. We take 
\[  I:=\Nat^{\omega,0} = \text{infinite sequences of natural numbers that are eventually $0$},\]
equipped with the (well-founded) lexicographical order 
\[a<b \qquad\text{iff}\qquad \text{$a\neq b$ and $a_m<b_m$ where $m=\max \{ i \mid a_i\neq b_i \}$}.\]
The \emph{rank} of a tree is the sequence $a\in \Nat^{\omega,0}$ where $a_i$ is the number of leaves of depth $i$. (The \emph{depth} of a node $x$ in a tree is the length of the unique path from the root to $x$.)
Note that if the tree $T$ has rank $a=(a_0,a_1,\ldots,a_{k-1},a_k,0,0,0,\ldots)$ where $a_k\neq 0$, then every successor tree $T'\in h(T)$ has a rank of the form
$a'=(a_0',a_1',\ldots,a'_{j-1},a_{j}-1,a_{j+1},\ldots ,a_k,0,0,0,\ldots)$, where $j$ is the depth of the deleted leaf, whence $a'<a$. This shows that the map $r\colon \Trees\to \Nat^{\omega,0}$ sending a tree to its rank is a ranking function for $(\Trees,h)$.
\end{proof}

By using a slightly different coalgebraic model, we can also capture the complexity (measured by the number of rounds until termination) of the Hydra game. We consider the functor $FX=(\Powf X)^\Nat$ and
the coalgebra 
\[ \ol{h}\colon \Trees \to (\Powf \Trees)^\Nat \]
such that $\ol{h}(T)(n)$ is the set of possible successor trees of $T$ in the Hydra game when the player chooses a pair of the form $(-,n)$. Note that the number of elements of $\ol{h}(T)(n)$ is bounded from above by the number of leaves of $T$, so it is a finite set.

\begin{proposition}\label{thm:hydra-recursive}
The coalgebra $(\Trees,\ol{h})$ is recursive.
\end{proposition}

\begin{proof}
The function $r\colon \Trees\to \Nat^{\omega,0}$ defined in the proof of \autoref{thm:hydra-well-founded} is also a ranking function for $(\Trees,\ol{h})$. Therefore this coalgebra is recursive by \autoref{thm:well-founded-recursive-coprod}.
\end{proof}
 A \emph{strategy} for the player is a sequence $\sigma\in \Nat^\omega$. We say that the player \emph{follows} $\sigma$ if in each round~$i$ before termination, a pair of the form $(l,\sigma_i)$ is chosen, that is, the number of leaves added to $l$'s grandparent (if any) is $\sigma_i$. The complexity of the Hydra game is measured by the function 
\[\maxsteps\colon \Trees \to \Nat^{\Nat^\omega}\] where $\maxsteps(T)(\sigma)$ is the maximum number of steps until termination that the game on $T$ takes when the player follows the strategy $\sigma$. To see that this maximum actually exists and thus the function $\maxsteps$ is total, observe that it adheres to the equation
\[
\maxsteps(T)(n:\sigma) = \max \{ \maxsteps(T')(\sigma) : \text{$T'$ successor of $T$ for any pair $(l,n)$} \}.
\]
Here $n : \sigma$ denotes the infinite sequence with head $n\in \Nat$ and tail $\sigma\in \Nat^\omega$. The above equation says precisely that $\maxsteps$ is a coalgebra-to-algebra morphism
for the $F$-algebra $m$ given by
\[ m(f)(n:\sigma)= \max \{ f_0(\sigma) \mid f_0\in f(n) \} \qquad \text{for $f\in (\Powf (\Nat^{\Nat^\omega}))^\Nat$, $n\in \Nat$, $\sigma\in \Nat^\omega$}.  \]
By recursivity of $\ol{h}$, the function $\maxsteps$ is total.

\takeout{
\subsection{Extended Euclidian Algorithm}
The extended Euclidian algorithm computes for any pair $(m,n)$ of integers a triple $\gcd(m,n)=(g,s,t)$ where $g$ is the greatest common divisor of $m$ and $n$ and $s$ and $t$ are integers such that $g=s\cdot m + t\cdot n$. The algorithm computes the function
\[ \gcd\colon \Nat\times\Nat \to \Nat\times \Int\times \Int \]  
recursively by
\[ \gcd(m,0) = (m,1,0) \]
and, for $n\neq 0$, 
\[ \gcd(m,n) = (g,t,s-tq)\qquad \text{where}\qquad (q,r)=(m \mathop{\mathrm{div}} n, m \mathop{\mathrm{mod}} n)\quad\text{and}\quad (g,s,t)=\gcd(n,r). \] 
As already observed by \citet{jeanninWellfoundedCoalgebrasRevisited2017}, this recursive definition says precisely that $\gcd$ is a coalgebra-to-algebra morphism
\[
\begin{tikzcd}
\Nat \ar{r}{\gcd} \ar{d}[swap]{\gamma} \times \Nat & \Nat\times \Int\times \Int \\
F(\Nat\times \Nat) \ar{r}{F\gcd} & F(\Nat\times\Int\times\Int) \ar{u}[swap]{\alpha}
\end{tikzcd}
\]
where $FX=\Nat+X\times\Nat$ on $\Set$ and the coalgebra $\gamma$ and algebra $\alpha$ are given by

\noindent\begin{minipage}{.5\textwidth}
\[
\gamma(m,n)=\begin{cases}
m & n=0\\
(n, m\mathop{\mathrm{mod}} n, m\mathop{\mathrm{div}}n) & n\neq 0
\end{cases}
\]
\end{minipage}
\begin{minipage}{.45\textwidth}
\begin{align*}
\alpha(g) &= (g,1,0)\\
\alpha((g,s,t),q)&=(g,t,s-tq)
\end{align*}
\vspace{.2cm}
\end{minipage}

\noindent To show that this diagram indeed defines $\gcd$ uniquely, we use our general recursivity result:

\begin{theorem}
The coalgebra $(\Nat\times\Nat,\gamma)$ is recursive.
\end{theorem}

\begin{proof}
By \autoref{thm:recursive-coprod} is suffices to show that this coalgebra is ranked. To this end, we use the well-founded relation 
\[
I=(\Nat\cup \{\top\},<)
\]  
given by the linear order of natural numbers with an additional top element (i.e.\ $n<\top$ for all $n\in\Nat$). Then the function
\[ r\colon \Nat\times \Nat\to \Nat\qquad\text{defined by}\qquad  r(m,n) = \begin{cases}
0 & n=0 \\  
  \top & m<n  \\
m & m\geq n>0 
\end{cases} \]   
is a ranking function for $(\Nat\times\Nat,\gamma)$.
\end{proof}
}

\section{Conclusion, Related Work, and Future Work}

We have presented a uniform foundation for proving intrinsic correctness of recursive algorithms, building on the key idea of capturing recursive branching by recursive coalgebras over indexed families. The technical centerpiece of our paper is the notion of {well-founded functors}, which allow designing those coalgebras in such a way that they are intrinsically recursive. We have shown through a number of fully formalized case studies that our coalgebraic framework is broadly applicable to various types of algorithms and provides a simple and principled approach to their formal verification.

Our approach bears some similarity to the Bove-Capretta method \cite{boveModellingGeneralRecursion2005}. The latter translates a generally recursive definition into an \emph{accessibility predicate}, which is somewhat similar to the base functor underlying our intrinsically recursive coalgebras over categories of indexed families. However, the accessibility predicate is used as an argument to a modified version of the original function. Also, it is indexed by the type of all the arguments to the original function, whereas our base functor can make dual use of indices already needed for intrinsic correctness. Exploring the precise formal connections between the two approaches is an interesting direction.

\citet[§~3.4]{castroperezProgramOptimisationsHylomorphisms2025} define a criterion for a specific class of coalgebras to be recursive.
They define this criterion specifically in the Rocq~\cite{herbelinRocqproverRocqRocq2026} prover for functors that can be modeled as containers~\cite{abbottContainersConstructingStrictly2005} and use setoids.
They focus on \emph{code extractability} and thus by design avoid indexing.
Our approach, in contrast, focuses on the indexed setting, allowing for synergies with intrinsic verification, allows proving recursivity for \emph{all} coalgebras of a given functor, and works for functors that are not containers.

\citet{schaeferIntrinsicVerificationParsers2025} define intrinsically correct parsers by working in \(\Set^{\mathsf{String}}\). In \autoref{sec:cyk-parsing}, we use the same correctness criteria as they do, however, our approach allows the definition of \emph{coalgebraic} parsing algorithms. We aim to investigate more such parsing algorithms
in future work.

Finally, we aim to lift the coalgebraic termination analysis techniques (\autoref{sec:ranked-coalgebras}) from $\Set$ to more general categories. This could be useful, for instance, for the analysis of algorithms that handle data, which are modelled by coalgebras over nominal sets. This direction is supported by recent advances in the theory of well-founded coalgebras, such as the coalgebraic K\H{o}nig's lemma~\cite{uw25}, and will require techniques such as canonical graphs and Ramsey's theorem beyond $\Set$.

\begin{acks}
Henning Urbat is supported by  \grantsponsor{018mejw64}{Deutsche Forschungsgemeinschaft (DFG, German
  Research Foundation)}{https://www.dfg.de} -- project numbers \grantnum{018mejw64}{470467389} and \grantnum{018mejw64}{569130867}.
We thank Max New and Steven Schaefer for a discussion about recursive coalgebras in the context of parsing.
\end{acks}

\section*{Data Availability Statement}
\label{sec:data-avail-stat}
The git development of our source code has been archived on Software Heritage \cite{swh-rev-b2dd125}. A reproducible build environment in which our formalization can be checked is available on Zenodo \cite{alexandruAgdaFormalizationIntrinsically2026}. The inline links to our formalization are to code compiled to html and hosted persistently as ancillary files to the {arXiv} version of our paper at \nicehref{\onlineHtmlURL index.html}{\onlineHtmlURL index.html}.
\printbibliography

@inproceedings{uw25,
  author       = {Thorsten Wi{\ss}mann and Henning Urbat},
  title        = {Well-Founded Coalgebras Meet {K\H{o}nig}’s Lemma},
  booktitle    = {Computer Science Logic ({CSL'26})},
  series       = {Leibniz International Proceedings in Informatics (LIPIcs)},
  publisher    = {Schloss Dagstuhl - Leibniz-Zentrum f{\"{u}}r Informatik},
  year         = 2026,
  volume       = 363,
  pages        = {24:1--24:19},
  month        = 2,
  doi          = {10.4230/LIPIcs.CSL.2026.24},
  preprinturl  = {https://www.arxiv.org/abs/2507.18539},
}

@book{adamek1990automata,
  author    = {Ad\'amek, Ji\v{r}\'i and Trnkov\'a, Věra},
  title     = {Automata and Algebras in Categories},
  series    = {Mathematics and Its Applications},
  volume    = {37},
  publisher = {Kluwer Academic Publishers},
  address   = {Dordrecht, Netherlands},
  year      = {1990},
  isbn      = {978-0-7923-0010-6}
}

@article{hoare62,
  author    = {C. A. R. Hoare},
  title     = {Quicksort},
  journal   = {The Computer Journal},
  volume    = {5},
  number    = {1},
  pages     = {10--16},
  year      = {1962},
  doi       = {10.1093/comjnl/5.1.10}
}

@book{knuth97,
  author    = {Donald E. Knuth},
  title     = {The Art of Computer Programming, Volume 2: Seminumerical Algorithms},
  edition   = {3rd},
  year      = {1997},
  publisher = {Addison-Wesley},
  address   = {Reading, Massachusetts},
  isbn      = {978-0-201-89684-8}
}

@book{hopcroft1979,
  title        = {Introduction to Automata Theory, Languages, and Computation},
  author       = {Hopcroft, John E. and Ullman, Jeffrey D.},
  year         = {1979},
  publisher    = {Addison-Wesley},
  address      = {Reading, MA},
  isbn         = {978-0201029888}
}

@book{awodey10,
  title     = {Category Theory},
  author    = {Awodey, Steve},
  edition   = {2},
  year      = {2010},
  series    = {Oxford Logic Guides},
  volume    = {52},
  publisher = {Oxford University Press},
}

@book{maclane98,
  title        = {Categories for the Working Mathematician},
  author       = {Mac Lane, Saunders},
  year         = {1998},
  edition      = {2},
  publisher    = {Springer},
  series       = {Graduate Texts in Mathematics},
  volume       = {5},
}

@book{amm25,
  author={Ad\'amek, Ji\v{r}\'i and Milius, Stefan and Moss, Lawrence S.},
  series={Cambridge Tracts in Theoretical Computer Science}, 
  title={Initial Algebras and Terminal Coalgebras: The Theory of Fixed Points of Functors},
  year={2025},
  publisher={Cambridge University Press},
  doi = {10.1017/9781108884112}
}

@INPROCEEDINGS{pr04,
  author={Podelski, Andreas and Rybalchenko, Andrey},
  booktitle={Proceedings of the 19th Annual IEEE Symposium on Logic in Computer Science, 2004.}, 
  title={Transition invariants}, 
  year={2004},
  volume={},
  number={},
  pages={32-41},
  doi={10.1109/LICS.2004.1319598}}

@article{KirbyParis1982,
  author    = {Laurie Kirby and Jeff Paris},
  title     = {Accessible independence results for {Peano} arithmetic},
  journal   = {Bulletin of the London Mathematical Society},
  volume    = {14},
  number    = {4},
  pages     = {285--293},
  year      = {1982},
  doi       = {10.1112/blms/14.4.285}
}

@inproceedings{AlexandruCRW25,
  author       = {Cass Alexandru and
                  Vikraman Choudhury and
                  Jurriaan Rot and
                  Niels van der Weide},
  title        = {Intrinsically Correct Sorting in Cubical {Agda}},
  booktitle    = {{CPP} 2025},
  oldbooktitle    = {Proc. International Conference on
                  Certified Programs and Proofs ({CPP} 2025)},
  pages        = {34--49},
  publisher    = {{ACM}},
  year         = {2025},
  url          = {https://doi.org/10.1145/3703595.3705873},
  doi          = {10.1145/3703595.3705873},
  bibsource    = {dblp computer science bibliography, https://dblp.org}
}

@article{CaprettaUV06,
  author       = {Venanzio Capretta and Tarmo Uustalu and Varmo Vene},
  title        = {Recursive coalgebras from comonads},
  journal      = {Inf. Comput.},
  volume       = 204,
  number       = 4,
  pages        = {437--468},
  year         = 2006,
  url          = {https://doi.org/10.1016/j.ic.2005.08.005},
  doi          = {10.1016/j.ic.2005.08.005},
  bibsource    = {dblp computer science bibliography, https://dblp.org},
}

@article{generalizeddeterminization,
  author       = {Alexandra Silva and Filippo Bonchi and Marcello M. Bonsangue and Jan J. M. M. Rutten},
  title        = {Generalizing determinization from automata to coalgebras},
  journal      = {Logical Methods in Computer Science},
  volume       = 9,
  number       = 1,
  year         = 2013,
  url          = {http://dx.doi.org/10.2168/LMCS-9(1:9)2013},
  doi          = {10.2168/LMCS-9(1:9)2013},
  bibsource    = {dblp computer science bibliography, http://dblp.org},
}

@article{boveModellingGeneralRecursion2005,
  title = {Modelling General Recursion in Type Theory},
  author = {Bove, Ana and Capretta, Venanzio},
  year = {2005},
  month = aug,
  journal = {Mathematical Structures in Computer Science},
  volume = {15},
  number = {4},
  pages = {671--708},
  issn = {1469-8072, 0960-1295},
  doi = {10.1017/S0960129505004822},
  url = {https://www.cambridge.org/core/journals/mathematical-structures-in-computer-science/article/abs/modelling-general-recursion-in-type-theory/B02BBE1A0A38C44189D53B01D659ECE5#},
  urldate = {2025-06-30},
  langid = {english}
}

@unpublished{eppendahlFixedPointObjects2000,
  title = {Fixed {{Point Objects Corresponding}} to {{Freyd Algebras}}},
  author = {Eppendahl, Adam},
  year = {2000},
  month = may,
  langid = {english}
}

@article{jeanninWellfoundedCoalgebrasRevisited2017,
  title = {Well-Founded Coalgebras, Revisited},
  author = {Jeannin, Jean-Baptiste and Kozen, Dexter and Silva, Alexandra},
  year = {2017},
  month = oct,
  journal = {Mathematical Structures in Computer Science},
  volume = {27},
  number = {7},
  pages = {1111--1131},
  issn = {0960-1295, 1469-8072},
  doi = {10.1017/S0960129515000481},
  url = {https://www.cambridge.org/core/journals/mathematical-structures-in-computer-science/article/wellfounded-coalgebras-revisited/05A7C960C18A88A65B321E4DDDA38EBB},
  urldate = {2024-05-27},
  langid = {english}
}

@article{schaeferIntrinsicVerificationParsers2025,
author = {Schaefer, Steven and Varner, Nathan and Azevedo de Amorim, Pedro Henrique and New, Max S.},
title = {Intrinsic Verification of Parsers and Formal Grammar Theory in Dependent {Lambek} Calculus},
year = {2025},
issue_date = {June 2025},
publisher = {Association for Computing Machinery},
address = {New York, NY, USA},
volume = {9},
number = {PLDI},
url = {https://doi.org/10.1145/3729281},
doi = {10.1145/3729281},
journal = {Proc. ACM Program. Lang.},
month = jun,
articleno = {178}
}

@article{cookPodelskiRybalchenko,
author = {Cook, Byron and Podelski, Andreas and Rybalchenko, Andrey},
title = {Proving program termination},
year = {2011},
issue_date = {May 2011},
publisher = {Association for Computing Machinery},
address = {New York, NY, USA},
volume = {54},
number = {5},
issn = {0001-0782},
url = {https://doi.org/10.1145/1941487.1941509},
doi = {10.1145/1941487.1941509},
journal = {Commun. ACM},
month = may,
pages = {88–98},
numpages = {11}
}

@book{taylor99,
  author={Paul Taylor},
  series={Cambridge Studies in Advanced Mathematics}, 
  title={Practical Foundations of Mathematics},
  year={1999},
  publisher={Cambridge University Press},
  doi = {10.2307/3621547}
}

@article{vezzosiCubicalAgdaDependently2019,
  title = {Cubical {Agda}: A Dependently Typed Programming Language with Univalence and Higher Inductive Types},
  shorttitle = {Cubical {Agda}},
  author = {Vezzosi, Andrea and M{\"o}rtberg, Anders and Abel, Andreas},
  year = {2019},
  month = jul,
  journal = {Proc. ACM Program. Lang.},
  volume = {3},
  number = {ICFP},
  pages = {87:1--87:29},
  doi = {10.1145/3341691},
  url = {https://dl.acm.org/doi/10.1145/3341691},
  urldate = {2024-08-28}
}

@inproceedings{leroyWellfoundedRecursionDone2024,
  title = {Well-Founded Recursion Done Right},
  booktitle = {{{CoqPL}} 2024: {{The Tenth International Workshop}} on {{Coq}} for {{Programming Languages}}},
  author = {Leroy, Xavier},
  year = 2024,
  month = jan,
  publisher = {ACM},
  address = {London, United Kingdom},
  HAL_ID = {hal-04356563},
  url = {https://inria.hal.science/hal-04356563},
}

@inproceedings{hinzeConjugateHylomorphismsMother2015,
  title = {Conjugate {{Hylomorphisms}} - {{Or}}: {{The Mother}} of {{All Structured Recursion Schemes}}},
  shorttitle = {Conjugate {{Hylomorphisms}} - {{Or}}},
  booktitle  = {{POPL}},
  author = {Hinze, Ralf and Wu, Nicolas and Gibbons, Jeremy},
  year = 2015,
  pages = {527--538},
  publisher = {ACM},
  doi = {10.1145/2676726.2676989}
}

@inproceedings{meijerFunctionalProgrammingBananas1991,
  title = {Functional {{Programming}} with {{Bananas}}, {{Lenses}}, {{Envelopes}} and {{Barbed Wire}}},
  booktitle = {{FPCA}},
  author = {Meijer, Erik and Fokkinga, Maarten M. and Paterson, Ross},
  year = 1991,
  series = {Lecture {{Notes}} in {{Computer Science}}},
  volume = {523},
  pages = {124--144},
  publisher = {Springer},
  doi = {10.1007/3540543961_7},
  url = {https://doi.org/10.1007/3540543961\_7},
  urldate = {2023-09-04}
}

@book{birdAlgebraProgramming1997,
  title = {Algebra of Programming},
  author = {Bird, Richard S. and de Moor, Oege},
  year = 1997,
  series = {Prentice {{Hall International}} Series in Computer Science},
  publisher = {Prentice Hall},
  isbn = {978-0-13-507245-5}
}

@article{osiusCategoricalSetTheory1974,
  title = {Categorical Set Theory: {{A}} Characterization of the Category of Sets},
  shorttitle = {Categorical Set Theory},
  author = {Osius, Gerhard},
  year = 1974,
  month = feb,
  journal = {Journal of Pure and Applied Algebra},
  volume = {4},
  number = {1},
  pages = {79--119},
  issn = {0022-4049},
  doi = {10.1016/0022-4049(74)90032-2},
  url = {https://www.sciencedirect.com/science/article/pii/0022404974900322},
  urldate = {2025-11-11}
}

@article{dybjerInductiveFamilies1994,
  title = {Inductive Families},
  author = {Dybjer, Peter},
  year = {1994},
  month = jul,
  journal = {Formal Aspects of Computing},
  volume = {6},
  number = {4},
  pages = {440--465},
  issn = {1433-299X},
  doi = {10.1007/BF01211308},
  url = {https://doi.org/10.1007/BF01211308},
  urldate = {2024-09-03},
  langid = {english}
}

@online{dotpatterns,
	title = {Function Definitions : Dot Patterns — {Agda} 2.8.0 documentation},
        author = {Agda Team, The},
	url = {https://agda.readthedocs.io/en/v2.8.0/language/function-definitions.html#dot-patterns},
	urldate = {2025-11-13},
        year = {2025}
}

@software{alexandruAgdaFormalizationIntrinsically2026,
  title = {{Agda} {{Formalization}} of "{{Intrinsically Correct Algorithms}} and {{Recursive Coalgebras}}"},
  author = {Alexandru, Cass and Urbat, Henning and Wi{\ss}mann, Thorsten},
  year = 2026,
  month = apr,
  doi = {10.5281/zenodo.19628996},
  url = {https://zenodo.org/records/19628996},
  publisher = {Zenodo}
}

@softwareversion{swh-rev-b2dd125,
    author = "Alexandru, Cass and Urbat, Henning and Wißmann, Thorsten",
    license = "CC-BY-4.0",
    date = "2026-03-17",
    year = "2026",
    month = mar,
    repository = "https://git8.cs.fau.de/software/intrinsically-recursive.git",
    title = {Agda Formalization of "Intrinsically Correct Algorithms and Recursive Coalgebras"},
    version = "v1",
    swhid = "swh:1:rev:b2dd1252276423aed08558e7e092989e96cc40f8;origin=https://git8.cs.fau.de/software/intrinsically-recursive;visit=swh:1:snp:418aedda2bf666298e2c2d1d45a6aa375c1f1eb6"
}

@software{alexandru_2024_14279034,
  author       = {Alexandru, Cass and
                  Choudhury, Vikraman and
                  Rot, Jurriaan and
                  van der Weide, Niels},
  title        = {Formalization accompanying the paper ``Intrinsically Correct Sorting in Cubical {Agda}''},
  month        = dec,
  year         = 2024,
  publisher = {Zenodo},
  version      = {v0.1},
  doi          = {10.5281/zenodo.14279034},
  url          = {https://doi.org/10.5281/zenodo.14279034}
}

@inproceedings{castroperezProgramOptimisationsHylomorphisms2025,
  title = {Program {{Optimisations}} via {{Hylomorphisms}} for {{Extraction}} of {{Executable Code}}},
  booktitle = {{{ITP}} 2025},
  author = {Castro Perez, David and Paviotti, Marco and Vollmer, Michael},
  year = 2025,
  series = {Leibniz {{International Proceedings}} in {{Informatics}} ({{LIPIcs}})},
  volume = {352},
  pages = {32:1--32:20},
  publisher = {Schloss Dagstuhl -- Leibniz-Zentrum f\"ur Informatik},
  address = {Dagstuhl, Germany},
  issn = {1868-8969},
  doi = {10.4230/LIPIcs.ITP.2025.32},
  url = {https://drops.dagstuhl.de/entities/document/10.4230/LIPIcs.ITP.2025.32},
  urldate = {2026-02-26},
  isbn = {978-3-95977-396-6}
}

@software{herbelinRocqproverRocqRocq2026,
  title = {Rocq-Prover/Rocq: {{Rocq}} 9.1.1},
  shorttitle = {Rocq-Prover/Rocq},
  author = {Herbelin, Hugo and P{\'e}drot, Pierre-Marie and {coqbot} and Gilbert, Ga{\"e}tan and D{\'e}n{\`e}s, Maxime and {letouzey} and Arias, Emilio Jes{\'u}s Gallego and Sozeau, Matthieu and Zimmermann, Th{\'e}o and Tassi, Enrico and Filliatre, Jean-Christophe and Roux, Pierre and Melquiond, Guillaume and Spiwack, Arnaud and Gross, Jason and {barras} and Boutillier, Pierre and Laporte, Vincent and Fehrle, Jim and Glondu, St{\'e}phane and Bertot, Yves and Hugunin, Jasper and {Pit-Claudel}, Cl{\'e}ment and Caglayan, Ali and {forestjulien} and Courtieu, Pierre and Besson, Fr{\'e}d{\'e}ric and Rousselin, Pierre and MSoegtropIMC and Leray, Yann},
  year = 2026,
  month = feb,
  doi = {10.5281/zenodo.18554917},
  url = {https://zenodo.org/records/18554917},
  publisher = {Zenodo}
}

@article{abbottContainersConstructingStrictly2005,
  title = {Containers: {{Constructing}} Strictly Positive Types},
  shorttitle = {Containers},
  author = {Abbott, Michael and Altenkirch, Thorsten and Ghani, Neil},
  year = 2005,
  month = sep,
  journal = {Theoretical Computer Science},
  series = {Applied {{Semantics}}: {{Selected Topics}}},
  volume = {342},
  number = {1},
  pages = {3--27},
  issn = {0304-3975},
  doi = {10.1016/j.tcs.2005.06.002},
  url = {https://www.sciencedirect.com/science/article/pii/S0304397505003373},
  urldate = {2026-03-13}
}

@misc{joyofcats,
  author       = {Jiří Adámek and Horst Herrlich and George E. Strecker},
  title        = {Abstract and Concrete Categories. The Joy of Cats},
  year         = 2004,
}
\label{submissionpagelast}

}{
}
\ifthenelse{\boolean{showappendix}}{
\appendix
\section{Preservation Properties for Recursive Coalgebras}\label{sec:rec-preservation}
We present two general results on recursive coalgebras regarding their preservation by coalgebraic constructions. These results will be instrumental for the proof of \autoref{thm:well-founded-recursive-coprod}.
The first result says that recursivity is preserved by natural transformations:
\begin{envcite}
\begin{lemma}[{\cite[Prop.~1]{eppendahlFixedPointObjects2000}}]\label{recursiveNat}
   For every natural transformation $\alpha\colon G\to F$
   and every recursive $G$-coalgebra $c\colon C\to GC$,
   the $F$-coalgebra $\alpha_C\cdot c\colon C\to FC$ is recursive.
\end{lemma}
\end{envcite}
The second preservation result asserts that recursivity is preserved by determinization.

\begin{notation}
We write $\Coalg(F)$ for the category of coalgebras for an endofunctor $F$. Its morphisms $h\colon (C,c)\to (D,d)$ are morphisms $h\colon C\to D$ of $\C$ such that $d\cdot h = Fh\cdot c$. 
\end{notation}

\begin{definition}[Determinization]\label{def:determinize}
  For an endofunctor $F\colon \C\to\C$ and an adjunction $L\dashv R\colon \C\to \D$, we define
  the \emph{determinization} functor $\bar L\colon \Coalg(RFL)\to \Coalg(F)$ via adjoint transposition:
  \[
    \begin{tikzcd}
      (C,c\colon C\to RFLC)\quad
      \arrow[mapsto]{r}{\bar L}
      & \quad(LC, c^\sharp\colon LC\to FLC).
    \end{tikzcd}
  \]
\end{definition}

\begin{remark}
  Determinization is closely related to the \emph{generalized powerset construction}~\cite{generalizeddeterminization}. The latter
  considers $GT$-coalgebras for a functor $G\colon \D\to\D$ and a monad
  $T\colon \D\to \D$ for which~$G$ lifts to the Eilenberg-Moore category of
  $T$. For the construction itself, it suffices to consider an arbitrary
  decomposition of the monad $T$ into adjoint functors $L\dashv R\colon \C\to \D$ as in \autoref{def:determinize}.
  The condition that the functor $G$ lifts means that there is some functor $F\colon \C\to \C$ such that 
  \[
    \begin{tikzcd}
      \C
      \arrow{r}{F}
      \arrow{d}[swap]{R}
      & \C
      \arrow{d}{R}
      \\
      \D
      \arrow{r}{G}
      & \D
    \end{tikzcd}
  \]
  commutes. Since
$
    GT
    = GRL
    = RFL,
$
  the generalized powerset construction is then the functor
  \[
    \Coalg(GT)
    = \Coalg(RFL)
    \xrightarrow{~~\bar L~~}
    \Coalg(F)
  \]
  as defined above.
  The familiar powerset construction for non-deterministic automata corresponds to the setting where $G X = 2\times X^A$ on $\D = \Set$, $T=\Pow$ (the powerset monad), and $\C$ is the category of complete semilattices (equivalently, algebras for the monad $\Pow$).
  The lifting of $G$ is given by $FX=2\times X^A$ where $2$ is the semilattice $0 \le 1$ and the
  semilattice structure on $2\times X^A$ is defined componentwise.
  Then $GT$-coalgebras are non-deterministic automata and $F$-coalgebras
  are deterministic automata whose state space carries a semilattice
  structure and whose transitions and final states respect that structure.
  For the present purposes, we will consider the case where $\D=\C^I$ and $L\dashv R$ is the adjunction $\coprod \dashv \Delta$ of coproducts
  for $I$-indexed families.
\end{remark}

Given that left adjoints preserve initial objects and that
recursivity of a coalgebra is in spirit some kind of initiality (namely with respect to algebras), the following result is not surprising:
\begin{lemma}\label{determinizationRecursive}
  For every functor $F\colon \C\to \C$  and every adjunction $L\dashv R\colon \C\to \C$, the 
  determinization functor $\bar L\colon \Coalg(RFL)\to \Coalg(F)$ preserves
  recursive coalgebras.
\end{lemma}
\begin{proof}
  Let $c\colon C\to RFLC$ be a recursive $RFL$-coalgebra. To prove that $c^\#\colon FC\to FLC$ is a recursive $F$-coalgebra, suppose that
  $a\colon FA\to A$ is an $F$-algebra. Using the counit $\epsilon\colon LR\to \Id_\C$ of the
  adjunction $L\dashv R$, we define the following $RFL$-algebra structure $a'$ on $RA$:
  \[
    a'\equiv\big(
      \begin{tikzcd}
        RFL(RA)
        \arrow{r}{RF\epsilon_A}
        & RF A
        \arrow{r}{Ra}
        & RA
      \end{tikzcd}
    \big).
  \]
  Since $(C,c)$ is recursive, we obtain a unique $h\colon C\to RA$ making the diagram on the left below commute. By adjointness, this is equivalent to commutativity of the diagram on the right.
  \[
    \begin{tikzcd}
      C
      \arrow{rr}{c}
      \arrow{d}[swap]{h}
      & & RFLC
      \arrow{d}{RFLh}
      \\
      RA
      &
      RFA
      \arrow{l}{Ra}
      & RFLRA
      \arrow{l}{RF\epsilon_A}
      \arrow[shiftarr={yshift=-6mm}]{ll}{a'}
    \end{tikzcd}
    \quad\overset{L\,\dashv\, R}{\Longleftrightarrow}\quad
    \begin{tikzcd}
      LC
      \arrow{rr}{c^\sharp}
      \arrow{d}[swap]{h^\sharp}
      & & FLC
      \arrow{d}{FLh}
      \\
      A
      &
      FA
      \arrow{l}{a}
      & FLRA
      \arrow{l}{F\epsilon_A}
    \end{tikzcd}
  \]
We show that the adjoint transpose $h^\#$ is the unique coalgebra-to-algebra morphism from $(LC,c^\#)$ to $(A,a)$. Indeed, since $h^\sharp=\epsilon_A\cdot Lh$,
  the second diagram above shows that $h^\sharp$ is a coalgebra-to-algebra morphism 
  ($h^\sharp = a\cdot F h^\sharp \cdot c^\sharp$).
  For uniqueness, suppose that $u\colon LC\to A$ is a coalgebra-to-algebra morphism ($u=a\cdot Fu\cdot c^\sharp$).
  Taking the adjoint transpose $u^\flat\colon C\to RA$ of $u$ yields by above the correspondence
  a coalgebra-to-algebra morphism from $(C,c)$ to $(RA,a')$, whence $u^\flat=h$
  by uniqueness of $h$, and so $u=h^\#$. This proves that the coalgebra $(FC,c^\#)$ is recursive.
\end{proof}

\ifthenelse{\boolean{proofsinappendix}}{%
\section{Omitted Proofs}
\closeoutputstream{proofstream}
\input{\maintexname-proofs.out}
}{}

\section{Index of Formalized Results}\label{agdarefsection}
Below we list the \textcolor{srcfilenamecolor}{Agda file} containing the referenced result and (if applicable) mention a concrete identifier in this file.
The respective HTML files and the Agda source code files can be found in the ancillary files on arxiv and on
\begin{center}
\nicehref{\onlineHtmlURL index.html}{\onlineHtmlURL index.html}
\end{center}
and are also directly linked below.
The code builds with Agda 2.8.0, cubical-0.9, and the agda standard library v2.3.
\printcoqreferences

}{
}

\end{document}